%% file: main_New.tex
\documentclass[11pt,a4paper]{article}
\usepackage[utf8]{inputenc}
\usepackage{geometry}
\usepackage[english]{babel}
\usepackage[dvipsnames]{xcolor}
\usepackage{amssymb}
\usepackage{amsmath}
\usepackage{amsfonts}
\usepackage{algorithm}
\usepackage{algorithmicx}
\usepackage{enumitem}
\usepackage{algpseudocode}
\usepackage{float}
\usepackage{graphics}
\usepackage{graphicx}
\usepackage{comment}
\usepackage{fontawesome5}

\usepackage{wrapfig}

\usepackage{multicol}

\usepackage{pdfpages}

\usepackage{latexsym} %Des symboles de maths en +
\usepackage{theorem} %pour changer les styles dans les theoremes.

\usepackage{thm-restate} %Réénoncer le thm exactement comme précédemment

\usepackage{hyperref}
\usepackage[capitalize]{cleveref}

\usepackage{tikz}
\usetikzlibrary{positioning} 
\usetikzlibrary{shapes}
\tikzstyle{Template_A}=[draw,circle,thick, minimum size = 0.5cm, color=black,inner sep=0pt]
\tikzstyle{Template_B}=[color=black,draw]
\tikzstyle{vertex}=[draw,circle,fill=black, color=black,inner sep=1.5pt]
\tikzstyle{vertex2}=[draw,circle,fill,inner sep=2pt]
\tikzstyle{vertex3}=[draw,circle,fill,inner sep=0.5pt]
\tikzstyle{S&L}=[draw,circle,minimum height=0.5cm]
\tikzstyle{labeled}=[draw,circle,inner sep=0.35pt]
\tikzstyle{boite}=[draw,rectangle,rounded corners=3pt]

\newtheorem{DE}{Definition}[section]

\newtheorem{theorem}[DE]{Theorem}
\newtheorem{lemma}[DE]{Lemma}

\newtheorem{remark}[DE]{Remark}
\newtheorem{property}[DE]{Property}

{\theoremstyle{break}\theorembodyfont{\rmfamily}}

\newcounter{claim}[theorem]
\newenvironment{proof}[1][]%
{\noindent{\it Proof. }{#1}{}}{\qed\vspace{2ex}}
\newenvironment{claim}[1][]%
{\refstepcounter{claim}\vspace{1ex} {(\it\arabic{claim}) {#1}{}}\it}{\vspace{2ex}}
\newenvironment{proofclaim}[1][]%
{\noindent {}{#1}}{This proves~(\arabic{claim}).\vspace{2ex}}

\newcommand {\sm} {\setminus}

\newcommand{\qed}{\relax\ifmmode\hskip2em\Box\else\unskip\nobreak\hfill$\Box$\fi}

\newcommand {\PC} {{\mathcal P}}

\newcommand\encircle[1]{%
  \tikz[baseline=(X.base)] 
    \node (X) [draw, shape=circle,inner sep=-0.3,outer sep=0.9ex] {\strut #1};
    }
    
\newcommand {\WC} {\ \encircle{w} \ }
\newcommand {\RT}{$(\star)$}

\author{Cléophée Robin\thanks{Normandie Univ, UNICAEN, ENSICAEN, CNRS, GREYC, 14000 Caen, France\\ Current affiliation: Université Paris Cité, CNRS, IRIF, F-75013, Paris, France} \hspace{0.3ex} and Eileen Robinson\thanks{Université libre de Bruxelles, Belgium}\\
{\small {\footnotesize\faEnvelope[regular]}~cleophee.robin@irif.fr \hspace{2ex} {\footnotesize\faEnvelope[regular]}~eileen.robinson@ulb.be}}

\title{Coloring bridge-free antiprismatic graphs}
\date{V1 August 2024, V2 February 2025}
\begin{document}

\maketitle

\begin{abstract}
The coloring problem is a well-researched topic and its complexity is known for several classes of graphs. However, the question of its complexity remains open for the class of antiprismatic graphs, which are the complement of prismatic graphs and one of the four remaining cases highlighted by Lozin and Malishev. In this article we focus on the equivalent question of the complexity of the clique cover problem in prismatic graphs. 

A graph $G$ is \textit{prismatic} if for every triangle $T$ of $G$, every vertex of $G$ not in $T$ has a unique neighbor in $T$. 
A graph is \textit{co-bridge-free} if it has no $C_4+2K_1$ as induced subgraph.
We give a polynomial time algorithm that solves the clique cover problem in co-bridge-free prismatic graphs. It relies on the structural description given by Chudnovsky and Seymour, and on later work of Preissmann, Robin and Trotignon.

We show that co-bridge-free prismatic graphs have a bounded number of disjoint triangles and that implies that the algorithm presented by Preissmann et al. applies. 
\end{abstract}

\section{Introduction}\label{s:intro}
In this article, we study the coloring problem for undirected simple graphs. A graph $G$ is defined as the couple of vertices $V(G)$ and edges $E(G)\subseteq V(G)\times V(G)$. We just use $V$ and $E$ if there is no ambiguity about the graph.
A \emph{$k$-coloring of a graph $G$} is a function $c:V(G)\rightarrow \{1, \dots, k\}$ such that for all $uv\in E(G)$, $c(u)\neq c(v)$. The coloring problem is the problem whose input is a graph $G$ and whose output is an integer $k$ such that  $G$ admits a $k$-coloring and $k$ is minimum. This problem is NP-hard and it is the object of much attention. Its complexity is refined for several classes of graphs.  

A \textit{clique} is a set of vertices that are pairwise adjacent. A \textit{stable set} is a set of vertices that are pairwise non-adjacent. The graphs $P_k$, $C_k$, $K_k$ are respectively the path, the cycle and the complete graph on $k$ vertices. The graph $K_{k, l}$ is the complete bipartite graph with one side of the bipartition of size $k$ and the other side of size $l$.
For two graphs $G$ and $H$ and an integer $k$, $G+H$ denotes the disjoint union of $G$ and $H$ and  $kG$ denotes the disjoint union of $k$ copies of $G$. 

For a given graph $H$, a graph $G$ is said to be \emph{$H$-free} if it does not contain any induced subgraph isomorphic to $H$. When $\cal H$ is a set of graphs, $G$ is \emph{$\cal H$-free} if $G$ is $H$-free for all
$H\in \cal H$. 

Král' et al.~\cite{daniel_kral_complexity_2001} proved the following dichotomy: the coloring problem for $H$-free graphs is polynomial time solvable if $H$ is an induced subgraph of $P_4$ or an induced subgraph
of $K_1+P_3$, and NP-hard otherwise. This motivated the systematic study of the coloring problem restricted to $\{H_1, H_2\}$-free graphs for all possible pairs of graphs $H_1$, $H_2$, or even to $\cal H$-free graphs in general.  The complexity has been shown to be polynomial or NP-hard for all sets of graphs on at most four vertices (see Figure~\ref{f:g4}), apart for the following four cases that are still open (see~\cite{golovach_coloring_2012} and \cite{v.v._lozin_vertex_2015}).

\begin{itemize}
\item ${\cal H = } \{K_{1, 3}, 4K_1\}$
\item ${\cal H = }\{K_{1, 3}, 2K_1+K_2\}$
\item ${\cal H = }\{K_{1, 3}, 2K_1+K_2, 4K_1\}$
\item ${\cal H = }\{C_4, 4K_1\}$
\end{itemize}

\begin{figure}
\centering
	\begin{tikzpicture}[scale=0.8]
	\node[vertex] (a) at (-5,1.5) {};
\node[vertex](b) at (-6,1.5) {};
\node[vertex] (c) at (-7,1.5) {};
\node[vertex] (d) at (-8,1.5) {};
\draw[-] (a) -- (b); 
\draw[-] (b) -- (c);
\draw[-] (c) -- (d); 
\node[-] (name) at (-6.5,1) {$P_4$}; 

\node[vertex] (a) at (0,0) {};
\node[vertex](b) at (0,-1) {};
\node[vertex] (c) at (-1,-1) {};
\node[vertex] (d) at (-1,0) {};
\draw[-] (a) -- (d); 
\draw[-] (b) -- (d);
\draw[-] (c) -- (d); 
\node[-] (name) at (-0.5,-1.5) {claw  $ = K_{1,3}$};

\node[vertex] (a) at (-3,0) {};
\node[vertex](b) at (-4,0) {};
\node[vertex] (c) at (-3,-1) {};
\node[vertex] (d) at (-4,-1) {};
\draw[-] (b) -- (c); 
\draw[-] (a) -- (b);  
\draw[-] (b) -- (d);
\draw[-] (c) -- (a);   
\node[-] (name) at (-3.5,-1.54) {paw};

\node[vertex] (a) at (-6,0) {};
\node[vertex](b) at (-7,-1) {};
\node[vertex] (c) at (-6,-1) {};
\node[vertex] (d) at (-7,0) {};
\draw[-] (a) -- (d);
\draw[-] (a) -- (c); 
\draw[-] (d) -- (b); 
\draw[-] (c) -- (b); 
\node[-] (name) at (-6.5,-1.5) {$C_4=K_{2,2}$};

\node[vertex] (a) at (-9,0) {};
\node[vertex](b) at (-10,-1) {};
\node[vertex] (c) at (-9,-1) {};
\node[vertex] (d) at (-10,0) {};
\draw[-] (a) -- (d);
\draw[-] (a) -- (c); 
\draw[-] (d) -- (b); 
\draw[-] (c) -- (b); 
\draw[-] (a) -- (b);
\node[-] (name) at (-9.5,-1.5) {diamond}; 

\node[vertex] (a) at (-12,0) {};
\node[vertex](b) at (-13,-1) {};
\node[vertex] (c) at (-12,-1) {};
\node[vertex] (d) at (-13,0) {};
\draw[-] (d) -- (b); 
\draw[-] (d) -- (c); 
\draw[-] (a) -- (d); 
\draw[-] (a) -- (b);
\draw[-] (a) -- (c);
\draw[-] (b) -- (c);   
\node[-] (name) at (-12.5,-1.5) {$K_4$};

\node[vertex] (a) at (0,-2.5) {};
\node[vertex](b) at (0,-3.5) {};
\node[vertex] (c) at (-1,-2.5) {};
\node[vertex] (d) at (-1,-3.5) {}; 
\draw[-] (b) -- (a);
\draw[-] (d) -- (a); 
\draw[-] (d) -- (b); 
\node[-] (name) at (-0.5,-4) {$K_3+K_1$}; 

\node[vertex] (a) at (-3,-2.5) {};
\node[vertex](b) at (-4,-2.5) {};
\node[vertex] (c) at (-3,-3.5) {};
\node[vertex] (d) at (-4,-3.5) {};
\draw[-] (b) -- (d); 
\draw[-] (a) -- (b);    
\node[-] (name) at (-3.5,-4) {$P_3+K_1$};

\node[vertex] (a) at (-6,-2.5) {};
\node[vertex](b) at (-7,-3.5) {};
\node[vertex] (c) at (-6,-3.5) {};
\node[vertex] (d) at (-7,-2.5) {};
\draw[-] (a) -- (c); 
\draw[-] (b) -- (d); 
\node[-] (name) at (-6.5,-4) {$2K_2$};

\node[vertex] (a) at (-9,-2.5) {};
\node[vertex](b) at (-10,-3.5) {};
\node[vertex] (c) at (-9,-3.5) {};
\node[vertex] (d) at (-10,-2.5) {};
\draw[-] (d) -- (c);
\node[-] (name) at (-9.5,-4) {$2K_1+K_2$}; 

\node[vertex] (a) at (-12,-2.5) {};
\node[vertex](b) at (-13,-3.5) {};
\node[vertex] (c) at (-12,-3.5) {};
\node[vertex] (d) at (-13,-2.5) {};
\node[-] (name) at (-12.5,-4) {$4K_1$};

\end{tikzpicture}	
\caption{All graphs on $4$ vertices\label{f:g4}}
\end{figure}

Lozin and Malyshev~\cite{v.v._lozin_vertex_2015} noted that a $\{K_{1, 3}, 2K_1+K_2\}$-free graph is either $4K_1$-free or has no edges. Therefore, $\{K_{1, 3}, 2K_1+K_2, 4K_1\}$-free graphs are essentially equivalent to $\{K_{1, 3}, 2K_1+K_2\}$-free graphs, in the sense that the complexity of the coloring problem is the same for both classes.

The $\{K_{1, 3}, 2K_1+K_2, 4K_1\}$-free graphs were first introduced in the context of claw-free\footnote{The \emph{claw} is another name for $K_{1, 3}$.} graphs by Chudnovsky and Seymour. The two first articles of the series (see~\cite{ChudAndSey1} and \cite{ChudAndSey2}) describe the so-called \emph{prismatic graphs} as the complement of a subclass of claw-free graphs.

A \emph{triangle} in a graph is a set of three pairwise adjacent vertices. A graph $G$ is \textit{prismatic} if for every triangle $T$ of $G$, every vertex of $G$ which is not in $T$ has a unique neighbor in $T$. 
Observe that if $\{s_1, s_2, s_3\}$ and $\{t_1, t_2, t_3\}$ are two vertex-disjoint triangles in a prismatic graph $G$, then there is a perfect matching in $G$ between $\{s_1, s_2, s_3\}$ and $\{t_1, t_2, t_3\}$, so that $\{s_1, s_2, s_3, t_1, t_2, t_3\}$ induces the complement of a $C_6$, commonly called a \emph{prism} (see Figure~\ref{f:prism}). This is where the name prismatic comes from.

\begin{wrapfigure}{r}{0.33\textwidth}
\begin{center}
\begin{tikzpicture}
\begin{scope}[xshift=0cm,yshift=0cm,scale=0.9]
		\node[-] (a) at (0,0) {$s_1$};
		\node[-] (b) at (0,2) {$s_2$};
		\node[-] (c) at (1,1) {$s_3$};
		\node[-] (d) at (4,0) {$t_1$};
		\node[-] (e) at (4,2) {$t_2$};
		\node[-] (f) at (3,1) {$t_3$};
		\draw[-] (a) -- (b);
		\draw[-] (a) -- (c);
		\draw[-] (c) -- (b);
		\draw[-] (d) -- (e);
		\draw[-] (d) -- (f);
		\draw[-] (e) -- (f);
		\draw[-] (a) -- (d);
		\draw[-] (b) -- (e);
		\draw[-] (c) -- (f);
\end{scope}
\end{tikzpicture}
\end{center}
\caption{Prism\label{f:prism}}
\end{wrapfigure}

A graph is \emph{antiprismatic} if its complement is prismatic.
It is straightforward to check that antiprismatic graphs are precisely $\{K_{1, 3}, 2K_1+K_2, 4K_1\}$-free graphs. 

Chudnovsky and Seymour gave a full structural description of prismatic graphs, and therefore of their complement.  They showed that the class can be divided into two subclasses to be defined later: the orientable prismatic graphs~\cite{ChudAndSey1}, and the non-orientable prismatic graphs~\cite{ChudAndSey2}.
It is natural to ask whether this description yields a polynomial time algorithm to color antiprismatic graphs.

The \emph{clique cover problem} is the problem of finding, in an input graph $G$, a minimum number of cliques that partition $V(G)$. It is equivalent to the coloring problem for the complement of the graph. It is therefore NP-complete in the general case. Our work is about the coloring problem for antiprismatic graphs.  However, it is more convenient to view it as a study of the clique cover problem for prismatic graphs.

A \textit{hitting set of the triangles} of a graph $G$ is a set of vertices intersecting every triangle of $G$. In this paper, we only consider hitting set of the triangles, therefore, for simplicity, we just say hitting set. We denote by $\Lambda(G)$ the minimum size of a hitting set of $G$.

Preissmann et al. proved in~\cite{preissmann_complexity_2021} that all non-orientable prismatic graphs can be clique-covered in $O(n^{7.5})$. To do so, they proved that every non-orientable prismatic graph either is a subgraph of the \textit{complement of the Schl\"afli graph} (to be defined later) or admits a hitting set of size at most $5$ (i.\,e.\ $\Lambda(G)\leq 5$). 

Such a result can not be stated for the whole class of prismatic graphs as the number of disjoint triangles in orientable prismatic graphs is unbounded. Informally, the existence of an orientation allows any given number of triangles to be “stacked” in a clean way. For a proper example of such structure, the reader can refer to the path and cycles of triangles graph introduced in~\cite{ChudAndSey1} and described again in section~\ref{s:Cycle_Path} and notice that the length  of the initial path or cycle is unbounded.

\begin{wrapfigure}{L}{0.5\textwidth}

\begin{tikzpicture}
\begin{scope}[xshift=0cm,yshift=0cm,scale=0.7]
		\node[vertex] (a) at (0.25,0) {};
		\node[vertex] (b) at (1,0) {};
		\node[vertex] (c) at (3,0) {};
		\node[vertex] (d) at (3.75,0) {};
		\node[vertex] (e) at (2,1.25) {};
		\node[vertex] (f) at (2,-1.25) {};
		\node[-] (name) at (2,-2.2) {Bridge};
		\draw[-] (e) -- (a);
		\draw[-] (e) -- (b);
		\draw[-] (e) -- (c);
		\draw[-] (e) -- (d);
		\draw[-] (a) -- (f);
		\draw[-] (b) -- (f);
		\draw[-] (d) -- (f);
		\draw[-] (c) -- (f);
		\draw[-] (e) -- (f);
		\draw[-] (a) -- (b);
		\draw[-] (c) -- (d);
\end{scope}

\begin{scope}[xshift=3.9cm,yshift=0cm,scale=0.7]
		\node[vertex] (b) at (0,-0.75) {};
		\node[vertex] (a) at (1.5,-0.75) {};
		\node[vertex] (c) at (0,0.75) {};
		\node[vertex] (d) at (1.5,0.75) {};
		\node[vertex] (e) at (2.75,-0.5) {};
		\node[vertex] (f) at (2.75,0.5) {};
		\node[-] (name) at (1.3,-2.2) {Co-bridge};
		\draw[-] (b) -- (a);
		\draw[-] (c) -- (b);
		\draw[-] (d) -- (c);
		\draw[-] (a) -- (d);
\end{scope}
\end{tikzpicture}
\caption{Bridge and Co-bridge}\label{f:bridge_et_co-bridge}
\end{wrapfigure}

Beineke~\cite{BEINEKE1970129} showed that the $K_{1,3}$ is one of the $9$ forbidden induced subgraphs for the characterisation of line graphs\footnote{$L(G)$ is the line graph of a graph $G$ if $V(L(G))=E(G)$, and $uv\in E(L(G))$ if and only if $u$ and $v$ share a common vertex in $G$.}. The coloring problem when restricted to line graphs is NP-Hard~\cite{Holyer1981TheNO} and various combinations of the $9$ forbidden induced subgraphs for line-graphs have been studied. Another forbidden induced subgraph for line graphs is the \textit{bridge} (see Figure~\ref{f:bridge_et_co-bridge}).

In this paper, we focus on the complement class of bridge-free antiprismatic graphs: the co-bridge-free prismatic graphs. These graphs are exactly the graphs in $\{K_3+K_1, \mbox{diamond}, K_4, C_4+K_2\}$-free.

Our main result is that Preissmann et al.'s statement is true for every co-bridge-free prismatic graph.

\begin{restatable}{theorem}{Main}\label{t:la_total}
If $G$ is a co-bridge-free prismatic graph then $G$ admits a hitting set of cardinality at most $5$ or $G$ is a Schl\"afli-prismatic graph. 
\end{restatable}

Using the same method as Preissmann et al.~\cite{preissmann_complexity_2021} we prove that co-bridge-free prismatic graphs can be clique-covered in $O(n^{7.5})$.

\subsection*{Outline}
In Section~\ref{s:Struct}, we give some definitions and pin a few key results to expose an overview of the structure of prismatic graphs introduced by Chudnovsky and Seymour.

The next three sections are devoted to show that subclasses of co-bridge-free orientable prismatic graphs admit a hitting set of size at most $5$ except in some very special cases. Section~\ref{s:Cycle_Path} focuses on path and cycle of triangles graphs, Section~\ref{s:3colo} on the 3-colorable graphs, and Section~\ref{s:non3colo} on the  not 3-colorable graphs.

Section~\ref{s:Lafin} is devoted to the proof of Theorem~\ref{t:la_total} and to a ${\cal O}(n^{7.5})$ algorithm for the clique covering problem of co-bridge-free prismatic graphs.

Section~\ref{s:conclusion} contains some concluding remarks.

%\section{Structure and Definitions}\label{s:Struct}
\section{Definitions and Structure of Prismatic Graphs}\label{s:Struct}

%\subsection*{Definitions}\label{ss:Def_not}
In this section, we expose the structure of orientable prismatic graphs given by Chudnovsky and Seymour in~\cite{ChudAndSey1}. Most definitions and notations for prismatic graphs are extracted from their work and we restate them for convenience. Some of the structures require more details and are therefore covered in dedicated sections. For other terms and notations not defined here, we rely on~\cite{BondyMurty}.

\subsection*{Graphs}
If $Y \subseteq V (G)$ and $x \in V (G) \setminus Y$, we say that $x$ is \textit{complete} (resp. \textit{anticomplete}) to $Y$ if $x$ is adjacent (resp. non-adjacent) to every member of $Y$. If $X, Y \subseteq V (G)$ are disjoint, we say that $X$ is \textit{complete} (resp. \textit{anticomplete}) to $Y$ if every vertex of $X$ is adjacent (resp. non-adjacent) to every vertex of $Y$. 

If $X, Y \subseteq V (G)$, we say that $X$ and $Y$ are \textit{matched} if $X \cap Y = \emptyset$, $|X| = |Y |$, and every vertex in $X$ has a unique neighbor in $Y$ and vice versa.

\subsection*{Prismatic graphs}

A graph $G$ is called \textit{prismatic} if for any triangle $T$, every vertex not in $T$ has exactly one neighbor in $T$. 

The structure of the class of prismatic graphs is easier to manipulate and provides an understanding of their complement, the class ${\cal H = }\{K_{1, 3}, 2K_1+K_2, 4K_1\}$ (see Section~\ref{s:intro}). 

Observing that the size of a hitting set is at least the number of vertex-disjoint triangles, the following property arises directly from the definition.

\begin{property}\label{p:HS1}
A prismatic graph $G$ contains two vertex-disjoint triangles if and only if it has no hitting set of size $1$, i.\,e.\ $\Lambda (G)\geq 2$. 
\end{property}

%This property follows directly from the definition of prismatic graphs and from the fact that the number of vertex disjoint triangles is always smaller or equal to the size of a hitting set. Therefore, we omit the proof. \ER{Doit-on préciser qu'on omet la preuve? } \ER{Proposition de phrase à mettre avant la prop: 

A prismatic graph $G$ is $k$\textit{-substantial} if for every $S \subseteq V (G)$ with $|S| < k$ there is a triangle $T$ with $S \cap T = \emptyset$. Hence if $G$ is not $k$-substantial then $\Lambda (G)\leq k$.

Chudnovsky and Seymour gave a structural description of prismatic graphs in~\cite{ChudAndSey1} and \cite{ChudAndSey2}. We rely on those results to provide the diagram in Figure~\ref{f:super_schema} in order to help the reader to understand the structure of this class of graphs.

\begin{figure}
\input{superschema}
\caption{Overview diagram of the structure of prismatic graphs\label{f:super_schema}}
\end{figure}

\subsection*{Orientable prismatic graphs}

The notion of orientation is the first distinction to divide prismatic graphs into two subclasses: the orientable prismatic graphs and the non-orientables prismatic graphs.

An \textit{orientation} $\cal O$ of a graph $G$ is a choice of a cyclic permutation ${\cal O}(T)$ for every triangle $T$ of $G$, such that if $S = \{s_1,s_2,s_3\}$ and $T = \{t_1,t_2,t_3\}$ are two disjoint triangles and $s_i t_i$ is an edge for $1\leq i\leq 3$, then ${\cal O}(S)$ is $s_1\rightarrow s_2 \rightarrow s_3 \rightarrow s_1$ if and only if ${\cal O}(T)$ is $t_1\rightarrow t_2 \rightarrow t_3 \rightarrow t_1$. We say that $G$ is \textit{orientable} if it admits an orientation.

Preissmann et al.~\cite{preissmann_complexity_2021} proved that every non-orientable prismatic graph either admits a hitting set of size at most $5$ or is a subgraph of some specific graph, the \textit{Schläfli
-prismatic graph}. Hence, we focus on orientable prismatic graphs. 

This class is considered in depth in~\cite{ChudAndSey1}, and it is divided into two subclasses %, the not 3-colorable prismatic graphs and the 3-colorable prismatic graphs 
according to the existence of a 3-coloring. 

\subsection*{3-colorable orientable prismatic graphs}
A 3-coloring of a graph $G$ is a partition of $V(G)$ into $3$ stable sets $A$, $B$, $C$. We call the quadruple $(G, A, B, C)$ a $3$-colored graph. A graph is called 3-colorable if it admits a 3-coloring and it is said not 3-colorable otherwise.

Chudnovsky and Seymour provided a structural theorem of 3-colored prismatic graphs with three basic classes, $\mathcal{Q}_0$, $\mathcal{Q}_1$, $\mathcal{Q}_2$, called \textit{prime graphs} which are described hereunder, and an operation called the \textit{worn chain decomposition} which is detailed in Section~\ref{s:3colo}.

\begin{theorem}[\cite{ChudAndSey1} 11.1]\label{tcs:prime}
    Every 3-colored prismatic graph admits a worn chain decomposition with all terms in $\mathcal{Q}_0\cup\mathcal{Q}_1\cup\mathcal{Q}_2$.
\end{theorem}

The class $\mathcal{Q}_0$ consists of all 3-colored graphs $(G, A, B, C)$ such that $G$ has no triangle.

The class $\mathcal{Q}_1$ corresponds to all 3-colored graphs $(G, A, B, C)$ where $G$ is isomorphic to $L(K_{3,3})$. This is the line graph of $K_{3,3}$ and it is presented in Figure~\ref{f:LK33}. Observe that every edge of $L(K_{3,3})$ belongs to a triangle.

The class $\mathcal{Q}_2$ contains all canonically-colored path of triangles graphs. Section~\ref{s:Cycle_Path} is dedicated to detail the more general structure of path of triangles graphs.% required to properly describe those graphs.

\begin{figure}
\begin{minipage}[c]{0.4\linewidth}\begin{center}
    \input{LK33}
    \end{center}
\caption{$L(K_{3,3})$ }\label{f:LK33}
\end{minipage}
\hfill
\begin{minipage}[c]{0.4\linewidth}\begin{center}
    \input{CoreRingFive}
    \end{center}
\caption{Core ring of five}\label{f:corering5}
\end{minipage}
\end{figure}

\subsection*{Non $3$-colorable orientable prismatic graphs}
If an orientable prismatic graph does not admit a 3-coloring, the following lemma of Chudnovsky and Seymour indicates precisely what this graph can be.

\begin{lemma}[\cite{ChudAndSey1} 11.2]\label{l:CetS_non3col}
    Every orientable prismatic graph that is not 3-colorable is either not 3-substantial, or a cycle of triangles graph, or a ring of five graph, or a mantled $L(K_{3,3})$.
\end{lemma}

Cycle of triangles graphs are structures related to path of triangles graphs and they are detailed together in Section~\ref{s:Cycle_Path}.
 
A graph $G$ is a \textit{mantled $L(K_{3,3})$} if $V(G)= A\cup V^1\cup V^2\cup V^3\cup V_1\cup V_2\cup V_3$ with $A=\{a^i_j : 1\leq i,j\leq 3\}$ and adjacencies as follows:
  \begin{itemize}[label=--]
    \item For $1\leq i,j,i',j' \leq3$, $a_j^i$ and $a_{j'}^{i'}$ are adjacent if and only if $i \neq i'$ and $ j \neq j'$. Observe that $G[A]$ is a $L(K_{3,3})$;
    \item For $1\leq i\leq 3$, $V_i$ and $V^i$ are stable sets;
    \item For $1\leq i\leq 3$, $V^i$ is complete to $\{a^i_1,a^i_2,a^i_3 \}$, and anticomplete to  $A\setminus \{a^i_1,a^i_2,a^i_3 \}$;
    \item For $1\leq i\leq 3$, $V_i$ is complete to $\{a_i^1, a_i^2,a_i^3 \}$ and anticomplete to $A\setminus \{a_i^1, a_i^2,a_i^3 \}$;
    \item $V_1 \cup V_2\cup V_3$ is anticomplete to $V^1 \cup V^2 \cup V^3$;
    \item There is no triangle included in $V_1 \cup V_2\cup V_3$ or $V^1 \cup V^2 \cup V^3$.
  \end{itemize}

\vspace{2ex}

%A core ring of five is the graph drawed in Figure~\ref{f:corering5}.

% Side by side figures 

A graph $G$ is a \textit{ring of five }if $V(G)=A\cup V_0\cup V_1\cup V_2\cup V_3\cup V_4\cup V_5$ with $A=\{a_1,\dots, a_5,b_1,\dots,b_5\}$. 

Let adjacency be as follows (reading subscripts modulo 5, where $a_i=a_5$ and $b_i=b_5$ when $i=0$.):
\begin{itemize}[label=--]
    \item For $1\leq i\leq 5$, $\{a_i,a_{i+1},b_{i+3}$\} is a triangle and $a_i$ is adjacent to $b_i$. Observe that $G[A]$ is the graph drawn in Figure~\ref{f:corering5}, called a \textit{core ring of five}; 
    \item $V_0$ is complete to $\{b_1,\dots,b_5\}$ and anticomplete to $\{a_1,\dots,a_5\}$;
    \item $V_0,V_1,\dots,V_5$ are all stable sets; 
    \item For $1\leq i\leq 5$, $V_i$ is complete to $\{a_{i-1},b_i,a_{i+1}\}$ and anticomplete to  $A\setminus \{a_{i-1},b_i,a_{i+1}\}$;
    \item $V_0$ is complete to $V_1\cup \dots\cup V_5$; 
    \item For $1\leq i\leq 5$, $V_i$ is anticomplete to $V_{i+2}$, and the adjacency between $V_i$, $V_{i+1}$ is arbitrary. 
\end{itemize}

\section{Cycle and Path of Triangles Graphs}\label{s:Cycle_Path}

Path of triangles graphs and cycle of triangles graphs are key classes of graphs in the description of orientable prismatic graphs and this section aims to show that they have a hitting set of size at most $5$. 
Both classes are defined by Chudnovsky and Seymour in~\cite{ChudAndSey1}. We merge their definitions in one set of conditions by emphasizing the difference between the cycle version with the modulo subscripts and the path version with a specific way to handle the extremities.

%For $G$ to be a \textit{cycle of triangles graph} we forces $n\geq 5$ with $n=2\mod 3$, $X_{2n+1}=X_1$ and we satisfy $(P1)-(P6)$ without primes.

%Let $G$ be a graph and pairwise disjoint stable subsets $X_1,\dots,X_{2n+1}$ of $V(G)$ with union $V(G)$. 

%For $G$ to be called a \emph{path of triangles graph}, it must satisfy the conditions $\mathcal{P}i$ for all $i\in \{1,\dots,6\}$.
%For $G$ to be called a \emph{cycle of triangles graph}, $n\geq 5$ and $n=2\mod 3$ and it must satisfy the following conditions $\mathcal{P}i$ for all $i\in \{1,\dots,6\}$ and $X_{2n+1}=X_1$.

Let us say that $G$ is a \textit{path of triangles graph} if for some integer $n\geq 1$ there is a partition of the vertex set into stable sets $X_1,\dots, X_{2n+1}$, satisfying the following conditions $\mathcal{P}1-\mathcal{P}7$.

Let us say that $G$ is a \textit{cycle of triangles graph} if for some integer $n\geq 5$ with $n=2$ modulo $3$, there is a partition of the vertex set into stable sets $X_1,\dots, X_{2n+1}$, satisfying the following conditions $\mathcal{P}1-\mathcal{P}6$, reading subscripts modulo $2n$.

\begin{itemize}
    \item[$\mathcal{P}1$] For $1\leq i \leq n$, there is a nonempty subset $\hat{X}_{2i} \subseteq  X_{2i}$ and at least one of $\hat{X}_{2i}$, $\hat{X}_{2i+2}$ has cardinality $1$.
    %\item[$\mathcal{P'}1$] Both $\hat{X}_2$ and $\hat{X}_{2n}$ have cardinality 1.
    
    \item[$\mathcal{P}2$] For $1 \leq i < j \leq 2n + 1$,
    \begin{enumerate}
         \item if $j-i = 2$ modulo $3$ and there exist $u\in X_i$ and $v \in X_j$ non-adjacent, then either $i$, $j$ are odd and $j = i + 2$, or $i$, $j$ are even and $u \notin \hat{X}_i$ and $v \notin \hat{X}_j$ ;
         \item if $j - i \neq 2$ modulo $3$ then either $j = i + 1$ or $X_i$ is anticomplete to $X_j$.
        \end{enumerate}
    %\item[$\mathcal{P'}2$] If $G$ is a cycle of triangles graphs, then $X_1$ is anticomplete to $X_{2n}$ Ok car on a $X_{2n+1}=X_1$.
    (Note that $k = 2 \mod 3$ if and only if $2n - k = 2 \mod 3$, so these statements are symmetric between $i$ and $j$.)
    
    \item[$\mathcal{P}3$] For $1\leq i \leq n + 1$, $X_{2i-1}$ is the union of three pairwise disjoint sets $L_{2i-1}$, $M_{2i-1}$, $R_{2i-1}$.
    %\item[$\mathcal{P'}3$] Moreover, $L_1=M_1=M_{2n+1}=R_{2n+1}=\emptyset$.
    
    \item[$\mathcal{P}4$] For $1\leq i \leq n$, $X_{2i}$ is anticomplete to $L_{2i-1} \cup R_{2i+1}$ ; $X_{2i} \setminus \hat{X}_{2i}$ is anticomplete to $M_{2i-1} \cup M_{2i+1}$ ; and every vertex in $X_{2i} \setminus \hat{X}_{2i}$ is adjacent to exactly one end of every edge between $R_{2i-1}$ and $L_{2i+1}$.
    
    \item[$\mathcal{P}5$] For $1\leq i \leq n$, if $|\hat{X}_{2i} | = 1$, then
    \begin{enumerate}
        \item $R_{2i-1}$, $L_{2i+1}$ are matched, and every edge between $M_{2i-1}\cup  R_{2i-1}$ and $L_{2i+1} \cup M_{2i+1}$ is between $R_{2i-1}$ and $L_{2i+1}$;
        \item the vertex in $\hat{X}_{2i}$ is complete to $R_{2i-1} \cup M_{2i-1} \cup L_{2i+1} \cup  M_{2i+1}$;
        \item $L_{2i-1}$ is complete to $X_{2i+1}$ and $X_{2i-1}$ is complete to $R_{2i+1}$;
        \item if $i>1$, $M_{2i-1}$, $\hat{X}_{2i-2}$ are matched, and if $i<n$, $M_{2i+1}$, $\hat{X}_{2i+2}$ are matched; if $M_1$ is not empty then $M_1$ is matched to $\hat{X}_2$ and $\hat{X}_{2n}$ (this is only true for cycle of triangles graphs thanks to $\PC 7.2$).
        %\item if $i=1$ or $i=n$, $\hat{X}_2$ is matched to $M_1$ if it is not empty. \Cleo{if $M_1$  is not empty then $M_1$ is matched to $\hat{X}_2$ and $\hat{X}_{2n}$ (this is only true for cycle of triangle graphs thanks to $\PC 7.2$)} %Condition only for cycle thanks to P7 (M1 is empty for path)
        \end{enumerate}
        
    \item[$\mathcal{P}6$] For $1 \leq i \leq n$, if $| \hat{X}_{2i} | > 1$, then 
    \begin{enumerate}
        \item $R_{2i-1} = L_{2i+1} = \emptyset$;
        \item if $u \in X_{2i-1}$ and $v \in X_{2i+1}$, then $u$, $v$ are non-adjacent if and only if they have the same neighbor in $\hat{X}_{2i}$.
        \end{enumerate}
        
    %\item[$\mathcal{P'}k$] If $R_1=\emptyset$ then $n\geq 2$ and $|\hat{X}_4|>1$, and if $L_{2n+1}=\emptyset$ then $n\geq 2$ and $|\hat{X}_{2n-2}|>1$. \ER{Je ne sais pas quel numéro lui mettre. 6? Ou 7? C'est le P7 dans les déf normales de C\&S}
    
    \item[$\mathcal{P}7$] For path of triangles graph only:
    \begin{enumerate}
        \item both $\hat{X}_2$ and $\hat{X}_{2n}$ have cardinality 1;
        \item $L_1=M_1=M_{2n+1}=R_{2n+1}=\emptyset$;
        \item if $R_1=\emptyset$ then $n\geq 2$ and $|\hat{X}_4|>1$, and if $L_{2n+1}=\emptyset$ then $n\geq 2$ and $|\hat{X}_{2n-2}|>1$. 
        \end{enumerate}

\end{itemize}

For a path of triangles graphs or cycle of triangles graphs, the partition of the vertices $X_1,\dots,X_{2n+1}$ as described by the previous conditions is called a \emph{good partition}.

The previous conditions on the subsets $X_i$ constrains a very specific structure on the graph. The three first ones, $\mathcal{P}1-\mathcal{P}3$, give a description of the general behavior of a set $X_i$ with any other set. 
The conditions $\mathcal{P}4-\mathcal{P}6$ give a description of the interactions between a set $X_i$ and its direct surrounding (i.\,e.\ $X_{i-1}\cup X_{i+1}$).
The condition $\mathcal{P}7$ gives a way to treat extremities of the path variations. This is not needed in the cycle variations as we read the subscripts modulo $2n$, therefore considering $X_{2n+1}= X_1$.

Moreover Chudnovsky and Seymour observe (\cite{ChudAndSey1}, remark before Lemma 3.1) that a vertex of $G$ is in no triangle of $G$ if and only if it belongs to one of the sets $X_{2i} \setminus \hat{X}_{2i}$. Hence, a vertex of $G$ is in some triangle if and only if it belongs to one of the sets $\hat{X}_{2i}$ or a set $X_{2i-1}$ for $1\leq i\leq n+1$. In the rest of the paper, we refer to this remark by \RT.

%\ER{Include this:}
%(Note that $k = 2 \mod 3$ if and only if $2n − k = 2 \mod 3$, so these statements are symmetric between i and j .) 
%Included directly on P2.
For a path of triangles graph or cycle of triangles graphs, the partition of the vertices $X_1,\dots,X_{2n+1}$ as described in the previous conditions is called a \emph{good partition}.

If $X_1,\dots, X_{2n+1}$ is a good partition of $G$, let $A_k = \bigcup(X_i : 1 \leq i \leq 2n + 1 \text{ and } i = k \mod 3) (k = 0, 1, 2)$.
It is shown in \cite{ChudAndSey1} that $(G, A_0, A_1, A_2)$ is a 3-colored graph. For any 3-colored graph $(G, A, B, C)$, if there is a good partition $X_1,\dots,X_{2n+1}$ of $G$ and sets $A_0, A_1, A_2$ as above, with ${A_0 , A_1 , A_2 }$ = ${A, B, C}$, we call $(G, A, B, C)$ a \textit{canonically-colored path of triangles graph}. This will be used in Section~\ref{s:3colo} when we discuss the case of 3-colorable prismatic graphs.

Before proving that co-bridge-free path of triangles graphs and co-bridge-free cycle of triangles graphs admit a hitting set of size at most 5, we need three useful structural lemmas based on the conditions P1-7. 

%We first prove two useful structural lemmas then we prove that if we ask for both structures to be co-bridge-free, they admit a hitting set of size at most 5. 

\begin{lemma}\label{l:Path_Cycle:postition_triangles_tout}
Let $G$ be a path of triangles graph or a cycle of triangles graph and let $T$ be a triangle in $G$. 

There exists $1\leq i\leq n$ such that either $T\subseteq R_{2i-1}\cup \hat{X}_{2i}\cup L_{2i+1}$ with $|\hat{X}_{2i}|=1$, or $T\subseteq \hat{X}_{2i-2}\cup M_{2i-1}\cup \hat{X}_{2i}$

(reading subscript modulo $2n$ when $G$ is a cycle of triangles graphs). 
\end{lemma}

\begin{proof}
Let $G$ be a path of triangles graph or a cycle of triangles graph with good partitions $X_1,\dots,X_\nu$ for $\nu=2n+1$ in the first case and $\nu=2n$ in the second case.  Let $T=\{a,b,c\}$ be a triangle in $G$.
 
Set $a\in X_i$, $b\in X_j$, $c\in X_k$ for some $i,j,k$. Recall that $X_i$, $X_j$ and $X_k$ are all stable sets and so $i\neq j\neq k\neq i$. Without loss of generality, assume that $1\leq i<j<k\leq \nu$ and so $i\neq k+1$. Since $ac\in E$, $k-i=2$ modulo $3$ by $\PC 2$.

Suppose now that $j-i=2$ modulo $3$. It follows that $k-j=(k-i)-(j-i)=2-2=0$ modulo $3$. Hence, by $\PC 2.2$ and as $bc\in E$, $k=j+1$ and $j-i=k-1-i=1$ modulo $3$, a contradiction. 
Therefore $j-i\neq 2$ modulo $3$ and by $\PC 2.2$ we have that $j=i+1$. It follows that $k-j=k-i-1=2-1=1$ modulo $3$ and by $\PC 2.2$, $k=j+1$. Therefore $a\in X_{j-1}$, $b\in X_{j}$, $c\in X_{j+1}$ with $1\leq j\leq 2n$. By \RT, if $j$ is even $T\subseteq X_{j-1}\cup \hat{X}_{j}\cup X_{j+1}$ and if $j$ is odd $T\subseteq \hat{X}_{j-1}\cup X_{j}\cup \hat{X}_{j+1}$.

%Recall that a vertex of $G$ is in no triangle of $G$ if and only if it belongs to one of the sets $X_{2i} \setminus \hat{X}_{2i}$. Hence, a vertex of $G$ is in some triangle if and only if it belongs to one of the sets $\hat{X}_{2i}$ or a set $X_{2i-1}$ for $1\leq i\leq n+1$. Hence, 

Suppose that $j$ is odd and so $a\in \hat{X}_{j-1}$ and $c\in \hat{X}_{j+1}$. By $\PC 4$, $\hat{X}_{j-1}$ is anticomplete to $R_j$ and $\hat{X}_{j+1}$ is anticomplete to $L_j$. By $\PC 3$, $X_j=M_j\cup L_j\cup R_j$ and so $b\in M_j$. Hence $T\subseteq \hat{X}_{2i-2}\cup M_{2i-1}\cup \hat{X}_{2i}$ for some $i\in\{1,\dots,n\}$ and the lemma holds. 

Suppose that $j$ is even and so $a\in \hat{X}_{j-1}$, $b\in \hat{X}_j$ and $c\in \hat{X}_{j+1}$. By $\PC 4$, $\hat{X}_{j}$ is anticomplete to $L_{j-1}\cup R_{j+1}$. Therefore $a\in M_{j-1}\cup R_{j-1}$ and $c\in M_{j+1}\cup L_{j+1}$. Now by $\PC 6.2$ if $|\hat{X}_j|>1$ then $a$ and $c$ would not be adjacent or would not have the same neighbor in $X_j$ by $\PC 6.2$. Therefore we have that $|\hat{X}_j|=1$. By  $\PC 5.1$, $a\in R_{j-1}$ and $c\in L_{j+1}$. Hence $T\subseteq R_{2i-1}\cup \hat{X}_{2i}\cup L_{2i+1}$ with $|\hat{X}_{2i}|=1$ and the lemma holds.
\end{proof}

\begin{lemma}\label{l:Path_Cycle:ens_non_vide}
Let $G$ be a path of triangles graphs or a cycle of triangles graph. For $i\in \{1,\dots,n-1\} $, the following hold:
\begin{itemize}
    \item $M_{2i+1}\neq \emptyset$;
    \item $|M_{2i+1}|=1$ if and only if $|\hat{X}_{2i}|=|\hat{X}_{2i+2}|=1$;
    \item if $G$ is a cycle of triangles graph, then $M_1\neq \emptyset$.
\end{itemize}
\end{lemma}

\begin{proof}
Let $G$ be a path of triangles graph or a cycle of triangles graph with good partitions $X_1,\dots,X_\nu$ for $\nu=2n+1$ in the first case and $\nu=2n$ in the second case. 

Let $i\in \{1,\dots,n-1\} $. By $\PC 1$, either $|\hat{X}_{2i}|=1$ or $|\hat{X}_{2i+2}|=1$. Hence,  by $\PC 5.4$, either $M_{2i+1}$ and $\hat{X}_{2i+2}$ are matched and so  $|M_{2i+1}|= |\hat{X}_{2i+2}|$, or $M_{2i+1}$ and $\hat{X}_{2i}$ are matched and so $|M_{2i+1}|= |\hat{X}_{2i}|$. In both cases $|M_{2i+1}|\geq 1$. Moreover $|M_{2i+1}|=1$ if and only if $|\hat{X}_{2i}|=|\hat{X}_{2i+2}|=1$. 
 
The case for $i=0$ is proved in the same way for a cycle of triangles graph by considering subscripts modulo $2n$.
\end{proof}

\begin{lemma}\label{l:X3X2n-2}
   If $G$ is a path of triangles graph with $n\geq2$, then $|M_3\cup L_3|\geq 2$ and $|M_{2n-1}\cup R_{2n-1}|\geq 2$. 
\end{lemma}

\begin{proof}
Both cases being symmetric, we only prove $|M_3\cup L_3|\geq 2$.

By Lemma~\ref{l:Path_Cycle:ens_non_vide} $|M_3|\geq 1$. If $|M_3|\geq 2$ then the result holds, hence, assume that $|M_3|=1$. By $\PC 7.1$, $|\hat{X}_2|=1$ and so, by $\PC 5.4$, $M_3$ is matched with $\hat{X}_4$. Therefore $|\hat{X}_4|=1$ and by $\PC 7.3$, $|R_1| \geq 1$. By $\PC 5.1$, $R_1$ and $L_3$ are matched. Therefore $|L_3|=|R_1|\geq 1$. Hence $|M_3\cup L_3|\geq 2$. %Similarly, $|M_{2n-1}\cup R_{2n-1}|\geq 2$. \ER{Faut-il repréciser ici? Je dirais soit au début, soit à la fin, mais pas les 2.}
\end{proof}

We now have all we need to prove the main results of this section. 

\begin{lemma}\label{l:path_cycle_hitting}
Let $G$ be a co-bridge-free prismatic graph. If $G$ is a path of triangles graph or a cycle of triangles graph, then $\Lambda(G)\leq 5$. 
\end{lemma}

\begin{proof}
Let $G$ be a co-bridge-free prismatic graph. Suppose that $G$ is a path of triangles graph or a cycle of triangles graph with good partitions $X_1,\dots,X_\nu$ for $\nu=2n+1$ in  the first case and $\nu=2n$ in the second case.  In the following, we read subscripts modulo $2n$ when $G$ is a cycle of triangles graph.

Observe that, by $\PC 1$, for all $i\in \{1,\dots,n\} $, $\hat{X}_{2i}\neq \emptyset$.

Set $S=\{v\in \hat{X}_{2i} \mid i=1,\dots, n\text{ and } |\hat{X}_{2i}|=1\}$. By $\PC 1$ at least one of  $\hat{X}_{2i}$ or $\hat{X}_{2i+2}$ has cardinality 1. Hence, by Lemma~\ref{l:Path_Cycle:postition_triangles_tout}, $S$ is a hitting set of $G$. Observe that $|S|\leq n$. Hence, we can assume that $n\geq 6$ for otherwise the lemma holds.%In addition, when $G$ is a cycle of triangles graph, $n=2$ modulo $3$ and so $n\geq 8$.

\begin{claim}\label{c:Path_Cycle:Pas_Deux_Pendants}
For $1\leq i< n$, if $L_{2i-1}\neq \emptyset$ and $R_{2i+1}\neq \emptyset$ then $|X_{2i+5}\cup X_{2i+8}|\leq 1$. 
\end{claim}

\begin{proofclaim}
First observe that when $G$ is a cycle of triangles graph $n=2$ modulo $3$. Therefore $n\geq 8$ and so the sets we consider do not intersect when reading subscripts modulo $2n$.

Suppose, by contradiction, that there exist $a\in L_{2i-1}$, $b\in R_{2i+1}$ and $c,d\in X_{2i+5}\cup X_{2i+8}$.
By $\PC 6.1$ we have $|\hat{X}_{2i+2}|=1$ and by Lemma~\ref{l:Path_Cycle:ens_non_vide}, $M_{2i+1}\neq \emptyset$. Set $e\in \hat{X}_{2i+2}$ and set $f\in M_{2i+1}$. By $\PC 5.2$, $be,fe\in E$.

%We prove now that $ab, af\in E$. If $|\hat{X}_{2i}|=1$ then, by $\PC 5.3$ we have $L_{2i-1}$ is complete to $X_{2i+1}$ and so $ab, af\in E$. If $|\hat{X}_{2i}|>1$ then $X_{2i}$ is anticomplete to $L_{2i-1}$ and $R_{2i+1}$ by $\PC 4$, and so $a$ has no neighbor in $\hat{X}_{2i}$. Therefore by $\PC 6.2$, $af, ab\in E$.
We prove now that $ab, af\in E$. By $\PC 4$,  $X_{2i}$ is anticomplete to $L_{2i-1}$ and $R_{2i+1}$. Therefore $a$ has no neighbor in $\hat{X}_{2i}$. If $|\hat{X}_{2i}|>1$ then $af, ab\in E$ by $\PC 6.2$. If $|\hat{X}_{2i}|=1$ then $L_{2i-1}$ is complete to $X_{2i+1}$ by $\PC 5.3$, and $ab, af\in E$.

Since $X_{2i+1}$ is a stable set, $bf\notin E$. By $\PC 2.2$, $ae\notin E$. Hence, $\{a,b,e,f\}$ induces a $C_4$.

Now, by $\PC 2.2$, $cd\notin E$ and $\{c,d\}$ is anticomplete to $\{a,b,e,f\}$. Hence $\{a,b,c,d,e,f\}$ induces a co-bridge, a contradiction. 
\end{proofclaim}

\begin{claim}\label{c:Path_Cycle:Pas_gros_ou_pas_long}
For $1\leq i \leq n$, if $|\hat{X}_{2i}|=1$ and $|M_{2i+1}\cup L_{2i+1}|\geq 2$ then $|X_{2i+10}\cup X_{2i+13}|\leq 1$.%\\(read subscripts modulo $2n$ when $G$ is a cycle of triangles graph). %\ER{Ou mettre 'read subscripts' plus haut}

%If $|\hat{X}_{2i}|=1$ and $|M_{2i-1}\cup R_{2i-1}|\geq 2$ then $|X_{2i-10}\cup X_{2i-13}|\leq 1$.
\end{claim}

\begin{proofclaim}
First observe that when $G$ is a cycle of triangles graph $n\geq 8$ and so the sets we consider do not intersect when reading subscripts modulo $2n$.

Suppose for the sake of a contradiction that there exist $a,b\in M_{2i+1}\cup L_{2i+1}$ and $c,d\in X_{2i+10}\cup X_{2i+13}$. By $\PC 2.2$ and since all sets $X_j$ are stable sets, $\{a,b,c,d\}$ is a stable set. 

Let $e$ be the only vertex in $\hat{X}_{2i}$. By $\PC 5.2$, $e$ is adjacent to $a$ and $b$. By $\PC 2.2$, $e$ is non-adjacent to $c$ and $d$. 

Let $f$ be a vertex in $\hat{X}_{2i+6}$. By $\PC 2.2$, $cf,df,ef\notin E$. By $\PC 2.1$, $af,bf\in E$. Hence, $\{a,b,c,d,e,f\}$ induces a co-bridge, a contradiction.
\end{proofclaim}

\begin{claim}\label{c:Path_Cycle:ML3}
If $|M_3\cup L_3|\geq 2$ then $\Lambda (G)\leq 5$.  
\end{claim}

\begin{proofclaim}
Suppose that $|M_3\cup L_3|\geq 2$. Observe first that we can assume that $|\hat{X}_2|=1$. By $\PC 7.1$ it is obvious when $G$ is a path of triangles graph. When $G$ is a cycle of triangles graph, if $|\hat{X}_2|>1$ then $|\hat{X}_0|=1$ by $\PC 1$ and $|M_1|\geq 2$ by Lemma~\ref{l:Path_Cycle:ens_non_vide}. Hence we can relabel all sets by adding $2$ to each label. Hence we can relabel all sets by updating indices $i$ such that $i\leftarrow i+2$ modulo ($2n+1$).

By (\ref{c:Path_Cycle:Pas_gros_ou_pas_long}) with $i=1$, $|X_{12}\cup X_{15}|\leq 1$. Hence $|X_{12}|\leq 1$,  $X_{15}=\emptyset$ and $n\leq 7$. 

If $G$ is a cycle of triangles graph, we already observed that it has a hitting set of size at most $5$ or $n\geq 8$. Hence we consider that $G$ is a path of triangles graph with $6\leq n\leq 7$.
%Just before, (\ref{c:Path_Cycle:Pas_Deux_Pendants}), we already assume that, when $G$ is a cycle of triangles graph, $n\geq 8$. Hence we can consider that $G$ is a path of triangles graph and $n\geq 6$. 

Suppose that $n=7$. Since $X_{15}=\emptyset$, $L_{15}=\emptyset$ and by $\PC 7.3$ $|\hat{X}_{12}|\geq 2$, a contradiction is obtained. Hence $n=6$. 

Recall that $S$ is the hitting set of $G$ defined earlier. If there exits at least one $1\leq i\leq n$ such that $|\hat{X}_{2i}|>1$ then $|S|\leq n-1$ and the claim holds. Hence, suppose that for all $1\leq i\leq n$, $|\hat{X}_{2i}|=1$. Observe that, by Lemma~\ref{l:Path_Cycle:ens_non_vide}, $|M_{2i+1}|=1$ for all $1\leq i\leq n-1$, and in particular $|M_3|=1$. Therefore $L_3\neq \emptyset$. 

We prove that $S\setminus \hat{X}_6$ is a hitting set of $G$. 

Suppose, by contradiction, that there exists a triangle $T=\{a,b,c\}$ not covered by $S\setminus X_6$. Since $S$ is a hitting set of $G$, $T\cap \hat{X}_6 \neq \emptyset$.  Since $\hat{X}_4\cup \hat{X}_8\subset S\setminus X_6$, we have $T\cap (\hat{X}_4\cup \hat{X}_8)=\emptyset$. Hence, by Lemma~\ref{l:Path_Cycle:postition_triangles_tout}, we have  $T \subset R_5\cup \hat{X}_6\cup L_7$. Since all sets are stable sets, it follows that $R_5\neq \emptyset$. 

Hence $L_3\neq \emptyset$ and $R_5\neq \emptyset$. By (\ref{c:Path_Cycle:Pas_Deux_Pendants}) (with $i=2$),  $|X_9\cup X_{12}|\leq 1$ and so $X_{12}=\emptyset$, a contradiction to $\PC 1$. Hence, such $T$ does not exist and $S\setminus \hat{X}_6$ is a hitting set of $G$ of size at most $5$. 

Hence, such $T$ does not exist and $S\setminus \hat{X}_6$ is a hitting set of $G$ of size at most $5$. 
\end{proofclaim}

By Lemma~\ref{l:X3X2n-2}, if $G$ is a path of triangles graph, then $|M_3\cup L_3|\geq 2$. By (\ref{c:Path_Cycle:ML3}), it follows that $\Lambda (G)\leq 5$ and the lemma holds.

If $G$ is a cycle of triangles graph then either $n= 8$ or $n\geq 11$ as $n=2$ modulo $3$.

Suppose that $n\geq 11$. By Lemma~\ref{l:Path_Cycle:ens_non_vide} and $\PC 1$, for all $1\leq i\leq 22$, $X_i\neq \emptyset$. Let $a\in X_1$, $b\in X_4$, $c\in X_{15}$, $d\in X_{18}$, $e\in X_8$ and $f\in X_{11}$. By $\PC 2$ we have that $\{a,b,c,d,e,f\}$ is a co-bridge with $\{a,b,c,d\}$ being a $C_4$, a contradiction. Hence $n=8$. 
 
Observe that, if there exists $0\leq i \leq 7$ such that $|M_{2i+1}\cup L_{2i+1}|>1$ then, up to relabeling, we can assume that $|M_3\cup L_3|>1$, and the lemma holds by (\ref{c:Path_Cycle:ML3}). Therefore suppose that $|M_{2i+1}\cup L_{2i+1}|=1$ for all $0\leq i \leq 7$. By Lemma~\ref{l:Path_Cycle:ens_non_vide}, it follows that $|M_{2i+1}|=1$, $ L_{2i+1}=\emptyset$ and $|\hat{X}_{2i+2}|=1$ and then by $\PC 5.1$, $R_{2i-1}$ and $L_{2i+1}$ are matched, so $R_{2i-1}=\emptyset$ for all $0\leq i \leq 7$. Now, by Lemma~\ref{l:Path_Cycle:postition_triangles_tout}, $S'=\hat{X}_2\cup \hat{X}_6 \cup \hat{X}_{10} \cup \hat{X}_{14}$ is a hitting set of size $4$.
\end{proof}

The following lemma states a more specific result for the hitting set on path of triangles graphs with at most $9$ vertices.  It will be helpful when we discuss the $3$-colorable prismatic graphs in Section~\ref{s:3colo}.

\begin{lemma}\label{l:Path_V9_hitting2V2}
   If $G$ is a path of triangles graph with $\vert V(G)\vert\leq 9$ then $\Lambda (G)\leq 2$.
\end{lemma}

\begin{proof}
    Let $G$ be a path of triangles graph with $|V(G)|\leq 9$. Let $X_1,\dots X_{2n+1}$ be a good partition of $G$.
    
    First observe that $| X_{2i}| \geq \vert\hat{X}_{2i}\vert\geq 1$ for $1\leq i\leq n$ by $\PC 1$. 
    
    By Lemma~\ref{l:Path_Cycle:postition_triangles_tout}, $S=\{v\in \hat{X}_{2i} \mid i=1,\dots, n\text{ and } |\hat{X}_{2i}|=1\}$ is a hitting set for $G$. Hence we can assume $n\geq 3$ for otherwise the lemma holds.

   By Lemma~\ref{l:X3X2n-2} and $\PC 3$, we have  $|X_3|\geq |M_3\cup L_3|\geq 2$ and $|X_{2n-1}|\geq |M_{2n-1}\cup R_{2n-1}|\geq 2$. In addition, by $\PC 7.3$ either $|X_1|\geq 1$ or $|X_4|\geq 2$ and either $|X_{2n+1}|\geq 1$ or $|X_{2n-2}|\geq 2$. Therefore $ |V(G)|=\sum_{1\leq j\leq 2n+1}|X_j|\geq 2n+3$. Since $|V(G)|\leq 9$ it follows that $n=3$. 
    
    At this point, if $|\hat{X}_4|\geq 2$ then by definition of $S$, $|S|\leq 2$ and the lemma holds. Hence assume that $|\hat{X}_4|=1$ and so, by $\PC 7.3$, $|X_1|\geq 1$ and $|X_7|\geq 1$ (as $2n-2=4$). 
    
    We show that $S'=S\sm \hat{X_4}$ is a hitting set of $G$. Since $S$ is a hitting set of $G$, by Lemma~\ref{l:Path_Cycle:postition_triangles_tout}, if a triangle is not covered by $S'$ then its vertices belong to $R_3\cup \hat{X}_4 \cup L_5$. Hence, $|R_3|\geq 1$ and $|L_5|\geq 1$.  Using previous observation, $|X_3|\geq |R_3\cup M_3\cup L_3|\geq 3$ and $|X_5|\geq |L_5\cup M_5\cup R_5|\geq 3$. It follows that $9 \geq \sum_{1\leq j\leq 2n+1}|X_j|\geq 2n+5$ and so $n\leq 2$, a contradiction.
\end{proof}

\section{$3$-Colorable Prismatic Graphs}\label{s:3colo}
The aim of this section is to prove the following: 

\begin{restatable}{theorem}{tTcolorable}\label{t:3_colorable}
If $G$ is a 3-colorable co-bridge-free prismatic graph then
$\Lambda(G)\leq 5$ or $G$ is a Schl\"afli-prismatic graph. 
\end{restatable}

Note that every $3$-colorable prismatic graph is orientable (shown in \cite{ChudAndSey1}~4.1). Hence, all results in this section holds for all $3$-colorable prismatic graphs as the notion of orientable is implied.

In order to do so, we recall that $3$-colorable prismatic graphs can be described as a worn chain decomposition into prime graphs as presented by Chudnovsky and Seymour \cite{ChudAndSey1}. We first show that every prime co-bridge-free $3$-colorable prismatic graph admits a hitting set of size at most $5$. Then we define the worn chain decomposition for a graph and show some interesting results about it. Afterwards we handle the exception for a 3-colorable co-bridge-free prismatic graph to have a hitting set of size at most $5$ by introducing the Schl\"afli-prismatic graphs. Finally, we prove Theorem~\ref{t:3_colorable}.

\subsection{Prime graphs}\label{ss:defPrime}
First recall that a 3-colorable prismatic graph is said to be prime if it is in $\mathcal{Q}_0\cup\mathcal{Q}_1\cup\mathcal{Q}_2$. The class $\mathcal{Q}_0$ consists of all 3-colored graphs with no triangle, the class $\mathcal{Q}_1$ corresponds to all 3-colored graphs which are isomorphic to $L(K_{3,3})$, and the class $\mathcal{Q}_2$ contains all canonically-colored path of triangles graphs (see Section~\ref{s:Cycle_Path}).

We are now ready to prove the following lemma:    
    
\begin{restatable}{lemma}{primecolorable} \label{l:3-col_prime}
Every prime 3-colorable co-bridge-free prismatic graph has a hitting set of size at most $5$.
\end{restatable}

\begin{proof}\label{proof_primecolorable}
    %Observe that Lemma~\ref{l:3-col_prime} is a direct corollary of Lemma~\ref{l:path_cycle_hitting}.

    Let $G$ be a prime 3-colorable co-bridge-free prismatic graph and $G\in \mathcal{Q}_0\cup \mathcal{Q}_1\cup \mathcal{Q}_2$. 

    If $G \in \mathcal{Q}_0$ then it has no triangle. Therefore $\Lambda(G)=0$.
    
    If $G \in \mathcal{Q}_1$, then it is isomorphic to $L(K_{3,3})$. Therefore $\Lambda (G)=3$, as $\{a_i^1, a_i ^2, a_i^3\}$ is a hitting set for any $i\in\{1,2,3\}$ (see Figure~\ref{f:LK33}). 

    If $G \in \mathcal{Q}_2$, then it is a canonically-colored path of triangles graph. Moreover it is co-bridge-free, therefore Lemma~\ref{l:path_cycle_hitting} ensures that $\Lambda(G)\leq 5$.
\end{proof}

%As prime graphs are the building blocks of any $3$-colorable orientable prismatic graphs of the class $\mathcal{Q}_0\cup\mathcal{Q}_1\cup\mathcal{Q}_2$, in particular of the canonically-colored path of triangles graphs. This lies in the next lemma. 

Chudnovsky and Seymour \cite{ChudAndSey1} proved that prime graphs are the building blocks of any $3$-colorable orientable prismatic graphs. In particular, the canonically-colored path of triangles graphs is the only class of prime graphs that could have hitting set larger than 3. Hence, to use canonically-colored path of triangles graphs as building block to non-prime $3$-colorable orientable prismatic graphs we need the next lemma. 

\begin{lemma}\label{l:canonicallycolored_hitting3ORcycle4}
    Let $(G,A,B,C)$ be a canonically-colored path of triangles graph. Then one of the following holds:
    \begin{itemize}
        \item $\Lambda(G)\leq3$;
        \item $\Lambda(G)=4$ and $G$ contains a bicolored $C_4$;
        \item $\Lambda(G)\geq 5$ and $G$ contains a bicolored $C_4+K_1$.
    \end{itemize}
\end{lemma}

\begin{proof}
    Set $S=\{v\in \hat{X}_{2i} \mid i=1,\dots, n\text{ and } |\hat{X}_{2i}|=1\}$. By $\PC 1$ at least one of  $\hat{X}_{2i}$ or $\hat{X}_{2i+2}$ has cardinality $1$. Hence, by Lemma~\ref{l:Path_Cycle:postition_triangles_tout}, $S$ is a hitting set of $G$. Observe that $|S|\leq n$. Hence, we can assume that $n\geq 4$ for otherwise the lemma holds.

    By Lemma~\ref{l:X3X2n-2}, $|M_3\cup L_3|\geq 2$. Set $x_2\in X_2$, $x_3, x_3'\in M_3\cup L_3$ and $x_8\in X_8$ ($X_8\neq \emptyset$ as $n\geq 4$). By $\PC 2$, $x_2x_8\notin E$ and $x_3x_8,x_3'x_8\in E$. Recall that $X_3$ is a stable set and so $x_3x_3'\notin E$. By $\PC 7.1$ $|\hat{X}_2|=1$ and so, by $\PC 5.2$, $x_2x_3,x_2x_3'\in E$. Hence ${\cal C}=\{x_2,x_3,x_3',x_8\}$ induces a bicolored $C_4$ with colors $A=A_0$ and $C=A_2$ (as in the definition of canonically-colored path of triangles graph in Section~\ref{s:Struct}).
    
    If $\Lambda(G)=4$, then the lemma holds. Suppose now that $\Lambda(G)\geq 5$ and so $n\geq 5$.
    
    By $\PC 2$, $X_{12}$ is anticomplete to $X_2$, $X_3$ and $X_8$. Hence if there exists $y\in X_{12}$ then $y$ has no neighbor in $\cal C$. By definition $y\in A_2=C$. Therefore $\{y\}\cup \cal C$ induces a bicolored $C_4+K_1$ and the lemma holds. 
    
    Suppose now that $X_{12}=\emptyset$. Then $n=5$ and we prove that there exists $y\in X_9$ such that $yx_8\notin E$. By Lemma~\ref{l:X3X2n-2}, $|M_9\cup R_9|\geq 2$ and by $\PC 7.1$, $|\hat{X}_{10}|=1$. If $|M_9|\geq 2$ then, by $\PC 5.4$, $\hat{X}_{8}$ is matched to $M_9$ and so, there exists $y\in M_9$ such that $yx_8\notin E$. If $|M_9|<2$ then $|R_9|\geq 1$.  By $\PC 4$, $X_{8}$ is anticomplete to $R_9$ and so there exists $y\in R_9$ such that $yx_8\notin E$. In both cases $y\in X_9$ with $yx_8\notin E$. Moreover $y$ has no neighbor in $\cal C$ by $\PC 2$ so ${\cal C}\cup \{y\}$ induces a $C_4+K_1$. 
    
    By the definition of canonically-colored path of triangles graph, $y\in X_9$ implies that $y\in A_0=A$ and so ${\cal C}\cup \{y\}$ is a bicolored $C_4+K_1$.
\end{proof}

\subsection{Worn chain decomposition}
By Theorem~\ref{tcs:prime} (11.1~\cite{ChudAndSey1}), we can restrict ourselves to the study of prime 3-colorable prismatic graphs and the worn chain decomposition. 

Let $n \geq 1$, and for $1 \leq i \leq n$, let
$(G_i , A_i , B_i , C_i )$ be a 3-colored prismatic graph, where $V (G_1), \dots, V (G_n)$ are all nonempty and pairwise vertex-disjoint. Let $A = A_1 \cup \dots\cup A_n$ , $B = B_1 \cup\dots\cup B_n$, and $C = C_1\cup \dots\cup C_n$ and let $G$ be the graph with vertex set $V (G_1)\cup\dots\cup V(G_n)$ and with adjacency as follows:
\begin{itemize}
    \item[$\mathcal{W}1$] For $1\leq i\leq n$, $G[ V(G_i)]=G_i$;
    \item[$\mathcal{W}2$] For $1 \leq i < j \leq n$, $A_i$ is anticomplete to $V(G_j)\setminus B_j$, $B_i$ is anticomplete to $V(G_j)\setminus C_j$, and $C_i$ is anticomplete to $V(G_j)\setminus A_j$; 
    \item[$\mathcal{W}3$] For $1\leq i <j\leq n$, if $u\in A_i$ and $v\in B_j$ are non-adjacent then $u$ and $v$ are both in no triangles, and the same applies if $u\in B_i$ and $v \in C_j$ and if $u \in C_i$ and $v\in A_j$.
\end{itemize}

In particular, $A$, $B$, $C$ are stable sets, and so $(G, A, B, C)$ is a 3-colored graph. We call the sequence
$(G_i , A_i , B_i , C_i )$ $(i = 1, . . . , n)$ a \textit{worn chain decomposition} or worn $n$-chain for $(G, A, B, C)$.

Note also that every triangle of $G$ is a triangle of one of $G_1,\dots,G_n$, and $G$ is prismatic. This is a remark from~\cite{ChudAndSey1} but we give here the proof.

\begin{property}\label{p:triangles in worn chain}
   Let $(G_i,A_i,B_i,C_i)$, with $1 \leq i \leq n$, be a worn $n$-chain of a 3-colored graph $(G,A,B,C)$.
   
   Every triangle of $G$ is a triangle of one of $G_1,\dots,G_n$, and $G$ is prismatic. 
\end{property}
\begin{proof}
In order to show the first assertion by contradiction, suppose that there exists a triangle $T=\{a,b,c\}$ such that $a\in V(G_i)$ and $b\in V(G_j)$ with $1\leq i<j\leq n$. Without loss of generality, assume that that $a\in A_i$. By $\mathcal{W}2$ and $\mathcal{W}3$, $b\in B_j$ as $ab\in E$. Since $A = A_1 \cup \dots\cup A_n$, $B = B_1 \cup\dots\cup B_n$ induce stable sets, $c\notin A\cup B$ and so $c\in C$. Let $k$ be the integer such that $c\in V(G_k)$. Since $ac\in E$, by $\mathcal{W}2$ and $\mathcal{W}3$, $k\leq i$ and so $k<j$. Since $bc\in E$, by $\mathcal{W}2$ and $\mathcal{W}3$, $k\geq j$, a contradiction.

To prove the second assertion, we need to show that for any triangle $T=\{a,b,c\}$ of $G$ and any vertex $v\in V(G)\setminus T$, $v$ has a unique neighbor in $T$. If there exits $1\leq i\leq n$ such that $T\cup \{v\}\in G_i$, then the result holds by $\mathcal{W}1$. Suppose that $T\in G_i$ and $v\in V(G_j)$ for some $1\leq i<j\leq n$ (the other case is symmetric). Since all colors are stable sets, we set $a\in A_i$, $b\in B_i $ and $c\in C_i $ without loss of generality. If $v\in A_j$ then $va,vc\notin E$ by $\mathcal{W}2$ and $vb\in E$ by $\mathcal{W}3$. If $v\in B_j$ then $va,vb\notin E$ by $\mathcal{W}2$ and $vc\in E$ by $\mathcal{W}3$. If $v\in C_j$ then $vb,vc\notin E$ by $\mathcal{W}2$ and $va\in E$ by $\mathcal{W}3$. In any case, $v$ has a unique neighbour in $T$ and the second assertion it proved.
\end{proof}

Chudnovsky and Seymour define a given 3-colorable prismatic graph $(G,A,B,C)$ with $V(G) \neq \emptyset$ to be prime if it cannot be expressed as a worn 2-chain. Moreover they show (~\cite{ChudAndSey1}, Lemma 14.1) that a prime graph has to be in $\mathcal{Q}_0\cup\mathcal{Q}_1\cup\mathcal{Q}_2$ as claimed in Subsection~\ref{ss:defPrime}.
%\vspace{2ex}\ER{pourquoi vspace?}

In the following, we denote by $G=G_1\WC  G_2$ when $G_1$ and $G_2$ form a worn 2-chain of $G$ and we denote by $G=\WC_{i=1}^n G_i$ when the sequence $(G_i , A_i , B_i , C_i )$ $(i = 1, . . . , n)$ is a worn chain decomposition for $G$.

%\Cleo{Suggestion de notation : Let $(G_i,A_i,B_i,C_i)$ and $(G_j,A_j,B_j,C_j)$ be two 3-colored prismatic graphs in some worn chain decomposition with $i<j$. Given a color $X \in \{A_i,B_i,C_i\}$, $X^{\rightarrow}\in \{A_j,B_j,C_j\}$ denote the color that is complete to $X$. Hence $A_i^{\rightarrow}=B_j$,  $B_i^{\rightarrow}=C_j$ and $C_i^{\rightarrow}=A_j$.  } \ER{Je ne crois pas que ce soit utile...}

For the sake of clarity, we exhibit the two following direct consequences of $\mathcal{W}2$ that we use frequently in the rest of this article.

Given $(A,B,C)$ a coloring of $G$ such that $G=\WC_{i=1}^nG_i$ then for any $X\in \{A,B,C\}$, $X_i=X\cap V(G_i)$.

\begin{property}\label{p:color}
   Let $(G_i,A_i,B_i,C_i)$, with $i=1,2$, be a worn 2-chain of a 3-colored graph $(G,A,B,C)$. Let $j,k\in\{1,2\}$ with $j\neq k$. For any two distinct colors $X,Y$ both properties hold:
   
   \begin{enumerate}
        \item [(1)]\label{p:color_anticomp} One of $X_k,Y_k$ is anticomplete to $X_j\cup Y_j$. 
        %\begin{proof}
            %Let $(G_i,A_i,B_i,C_i)$, with $i=1,2$, be a worn 2-chain decomposition for a graph $(G,A,B,C)$. Let $j,k\in\{1,2\}$ with $j\neq k$ and let $X_j,Y_j$ be two colors of $G_j$.
            %  First recall that $X=X_j\cup X_k$ and $Y=Y_j\cup Y_k$ are both stable sets. 
            %  Suppose that $j<k$. By the cyclicity of $\mathcal{W}2$, either $X_j$ is anticomplete to $Y_k$ or $Y_j$ is anticomplete to $X_k$. Therefore one of $X_k, Y_k$ is anticomplete to $X_j\cup Y_j$. The case $k<j$ holds by symmetry. 
        %\end{proof}
        \item [(2)]\label{p:color_comp} Either all vertices of $X_k$ in a triangle are complete to $Y_j$ or all vertices of $Y_k$ in a triangle are complete to $X_j$.
        %\item \Cleo{If  $u\in X_k$ and  $v\in Y_k$ such that $uv\notin E$ then both $u$ and $v$ are in no triangle.}

        %\begin{proof}
        %    This holds by cyclicity of $\mathcal{W}3$. 
        %\end{proof}
       
   \end{enumerate}
\end{property}

\begin{remark}\label{r:color_triangles_disjoints}
    If $(G,A,B,C)$ is a 3-colored prismatic graph with at least $k$ vertex-disjoint triangles, then $|A|,|B|,|C|\geq k$ as each triangle contains a vertex of each color. 
\end{remark}

%Therefore we will say that a 3-colorable prismatic graph is prime if $V(G)\neq \emptyse$ and $G\in \mathcal{Q}_i$ for a certain $i\in \{1,2,3\}$.

%Towards showing that every 3-colorable orientable co-bridge-free prismatic graph has a hitting set of bounded size, we analyze worn 2-chains of a graph. 

%We show a two lemmas one that exhibit a bicolored $P_3$ a,d the other that bound the number of vertices on one part of the decomposition.\Cleo{Je suis d'avis de faire sauter cette phrase. Aussi parce qu'il n'y a qu'un seul resultat qui est sur la worn chain. : 

Towards showing that every $3$-colorable co-bridge-free prismatic graph has a hitting set of bounded size, we need some intermediate results. 
%}

%In order to show that every 3-colorable orientable co-bridge-free prismatic graph has a hitting set of bounded size, we analyze the worn chain decomposition for the graph with all terms in $\mathcal{Q}_0\cup \mathcal{Q}_1\cup \mathcal{Q}_2$. In particular, we need to unsure that the number of elements in $\mathcal{Q}_1$ and $\mathcal{Q}_2$ of a graph in this class is bounded, which we show in the following lemmas.

\begin{lemma}\label{l:P3bicolored}
    Let $G$ be a 3-colorable prismatic graph with at least two disjoint triangles. For any couple of colors $\{X,Y\}$ in a 3-coloring of $G$, there exist $3$ vertices $x_1, x_2, x_3$ which induce a $P_3$ and $x_1,x_3\in X$, $x_2\in Y$.
\end{lemma}

\begin{proof}
    Let $T_1, T_2$ be two disjoint triangles, let $(G, A, B, C)$ be a 3-colored prismatic graph and let $\{X,Y\}$ be any two different colors in $\{A,B,C\}$. 
    Let $x_i$ (resp. $y_i$) be the vertex in $T_i$ with color $X$ (resp. $Y$). If $x_1y_1x_2$ is an induced path, we are done. If it is not, then $x_2y_1$ is not an edge and as the graph is prismatic and 3-colored, $x_1y_2$ is an edge which yields $x_1y_2x_2$ to be the sought induced path.
\end{proof}

\begin{lemma}\label{l:prism_worn_V9_NoTriangle}
    Let $G$ be a $3$-colorable co-bridge-free prismatic graph. If it admits a worn 2-chain decomposition where one of the graphs has two disjoints triangles, then the other graph has either at most $9$ vertices or it has no triangle.
\end{lemma}

\begin{proof}
    %Let $(A,B,C)$ be a 3-coloring for $G$ and let $G_1, G_2$ be two 3-colored graphs such that $G=G_1\WC G_2$.
    %Set $i\neq j \in \{1,2\}$ such that $G_i$ is a graph with two disjoint triangles. For the sake of a contradiction, assume that $G_j$ has at least 10 vertices and a triangle $T$.
    
    %Since $G_j$ admits a 3-coloring, the pigeonhole principle ensures that there exist at least $4$ vertices with the same color $X_j$ with $X\in \{A,B,C\}$.
    
    %There is exactly one vertex $x\in X_j$ which is also in the triangle $T$. As $X_j$ is a stable set and $G_j$ is prismatic, every other vertex in $X_j$ has exactly one neighbor in $T\setminus \{x\}$. By the pigeonhole principle, there are two vertices $u,v\in X_j$ which are not adjacent to one of the vertices of $T\setminus \{x\}$. Let $y$ be the vertex in $T$ such that $uy,vy\notin E(G)$, and let $Y_j$ be its color with $Y\in \{A,B,C\}\sm X$.    
    
    %By Property~\ref{p:color}~(1), one of $X_i, Y_i$ is anticomplete to $X_j\cup Y_j$. Without loss of generality, suppose $X_i$ is this color. Then $Y_i$ is anticomplete to $Y_j$ and complete to any vertex of $X_j$ in a triangle. 
    %Moreover, as there are two disjoint triangles in $G_i$, Lemma~\ref{l:P3bicolored} ensure the existence of a path $\mathcal{P}=abc$ in $G_i$ with $\{a,c\}\in X_i$ and $b\in Y_i$. Therefore $\{u,v\}$, which is a stable set, is anticomplete to $\mathcal{P}\cup y$, which is a $C_4$. Hence $G$ has a co-bridge, a contradiction.

Let $(A,B,C)$ be a 3-coloring for $G$ and let $G_1, G_2$ be two 3-colored graphs such that $G=G_1\WC G_2$.
    Set $i\neq j \in \{1,2\}$ such that $G_i$ is a graph with two disjoint triangles. For the sake of a contradiction, assume that $G_j$ has at least 10 vertices and a triangle $T$.
    
    Since $G_j$ admits a 3-coloring, the pigeonhole principle ensures that there exist at least $4$ vertices with the same color, without loss of generality, let $A_j$ be that color.
    
  There is exactly one vertex $x\in A_j$ which is also in the triangle $T$. As $A_j$ is a stable set and $G_j$ is prismatic, every other vertex in $A_j$ has exactly one neighbor in $T\setminus \{x\}$. By the pigeonhole principle, there are two vertices $u,v\in A_j$ which are not adjacent to one of the vertices of $T\setminus \{x\}$. Let $y$ be the vertex in $T$ such that $uy,vy\notin E(G)$, and let $\Sigma_j$ be the color of $y$ with $\Sigma \in \{B,C\}$.

     By $\mathcal{WC} 2$ and $\mathcal{WC} 3$ there is a unique $X\in \{A,B,C\}$ such that $\Sigma$ is complete to all vertices in $X_i$ that are in a triangle and $A_j$ is anticomplete to $X_i$. In addition, $A_j$ is anticomplete to $X_i$. Again by $\mathcal{WC} 2$ and $\mathcal{WC} 3$, there is a unique $Y\in \{A,B,C\}$ such that $Y_i$ is anticomplete to both $A_j$ and $\Sigma_j$.  Moreover, as there are two disjoint triangles in $G_i$, Lemma~\ref{l:P3bicolored} ensure the existence of a path $\mathcal{P}=abc$ in $G_i$ with $\{a,c\}\in X_i$ and $b\in Y_i$. Hence, since $y\in \Sigma$, $\mathcal{P}\cup y$ induces a $C_4$. Also $X_i\cup Y_i$ is anticomplete to $A_j$ and $A$ is a stable set. Hence $\mathcal{P}\cup \{y\}\cup \{u,v\}$ induces a co-bridge, a contradiction. 
\end{proof}

\subsection{Schl\"afli-prismatic graphs}

Before proving Theorem~\ref{t:3_colorable} we need to study a special case of non-prime 3-colorable prismatic graphs. To do so we introduce the Schl\"afli-prismatic graphs which were introduced for the non-orientable prismatic graphs in~\cite{ChudAndSey2}.

The \emph{complement of the Schl\"afli graph (see figure~\ref{f:schafli}})
has 27 vertices $r^i_j,s^i_j,t^i_j$, $1\leq i,j \leq 3$ with
adjacencies as follows. For $1 \leq i, i', j, j' \leq 3$,
\begin{itemize}[label=--]
\item if $i\neq i'$ and $j\neq j'$, then $r^i_j$ is adjacent to $r^{i'}_{j'}$, $s^i_j$ is adjacent to $s^{i'}_{j'}$ and $t^i_j$ is adjacent to $t^{i'}_{j'}$;
\item if $j=i'$, then $r^i_j$ is adjacent to $s^{i'}_{j'}$, $s^i_j$ is adjacent to $t^{i'}_{j'}$ and $t^i_j$ is adjacent to $r^{i'}_{j'}$;
\item there are no other edges.
\end{itemize}

\input{Schlafli}

The class of \textit{Schläfli-prismatic} is the class of all induced subgraphs of the complement of the Schl\"afli graph and they are all prismatic.
%All induced subgraphs of the complement of the Schl\"afli graph are prismatic, and we call any such graphs \textit{Schläfli-prismatic}.
A smallest hitting set of the complement of the Schl\"afli graph has cardinality $10$ (\cite{preissmann_complexity_2021}~Lemma 3.6) but this graph non-orientable.
Hence, by considering only orientable prismatic graphs we restrict ourselves to subgraphs of the complement of the Schl\"afli graph (with less vertices and smaller hitting set).

\begin{lemma}\label{l:LK33wLK33schlafli}
    Let $G=G_1\WC G_2$ be a $3$-colorable  co-bridge-free prismatic graph. If both $G_1$ and $G_2$ contain an $L(K_{3,3})$ graph then $G$ is a Schl\"afli-prismatic graph and has a hitting set of size at most $6$.
\end{lemma}

\begin{proof}
   Let $G=G_1\WC G_2$ be a $3$-colorable co-bridge-free prismatic graph such that both $G_1$ and $G_2$ contain a $L(K_{3,3})$ graph. 
    
    By Lemma~\ref{l:prism_worn_V9_NoTriangle}, both $G_1$ and $G_2$ have at most $9$ vertices so they are both exactly isomorphic to $L(K_{3,3})$. 
    
   It is easy to check that the subgraphs induced by $R=\{r^i_j : 1 \leq i, i' , j, j'\leq3\}$, by  $S=\{s^i_j : 1 \leq i, j\leq3\}$ or by $T=\{t^i_j : 1 \leq i, j\leq3\}$ in the definition of the complement of the Schl\"afli graph are all $L(K_{3,3})$.
    
    Using same notations as in Figure~\ref{f:LK33}, one can see that, in $L(K_{3,3})$, the colors of a 3-coloring are either $\{a^1_1,a^1_2,a^1_3\},\{a^2_1,a^2_2,a^2_3\},\{a^3_1,a^3_2,a^3_3\}$ or $\{a^1_1,a^2_1,a^3_1\},\{a^1_2,a^2_2,a^3_2\},\{a^1_3,a^2_3,a^3_3\}$. In addition, both are symmetric. 
    
    With $1\leq i,j \leq 3$, denote by $u^i_j$ the vertices of $G_1$ and $v^i_j$ the vertices of $G_2$ with subscript as in the definition of $L(K_{3,3})$. 
    
    By the previous remark, assume that $A_1=\{u^1_1,u^2_1,u^3_1\}$, $B_1=\{u^1_2,u^2_2,u^3_2\}$ and $C_1=\{u^1_3,u^2_3,u^3_3\}$. 
    
    Up to relabeling, we can assume that $A_2=\{v^1_1,v^1_2,v^1_3\}$, $B_2=\{v^2_1,v^2_2,v^2_3\}$ and $C_2=\{v^3_1,v^3_2,v^3_3\}$.
    
    Observe that $G$ is isomorphic to the complement of the Schl\"afli graph induced by $R\cup S$.  Indeed, $G_1$ is isomorphic to the  complement of the Schl\"afli graph induced by $R$ and $G_2$ is isomorphic to the  complement of the Schl\"afli graph induced by $S$. In addition, all edges given by $\mathcal{W}2$ in the definition of a worn chain decomposition are exactly the one given by the second adjacency rule of the complement of the Schl\"afli graph.
  
    Therefore $G$ is a Schl\"afli-prismatic graph. Moreover, it has a hitting set of size $6$. Indeed, by Property~\ref{p:triangles in worn chain}, all the triangles are in $G_1$ and $G_2$ which are isomorphic to $L(K_{3,3})$ graphs and admit a hitting set of size $3$ each.
\end{proof}

\subsection{Proof of Theorem~\ref{t:3_colorable}}
We now have the main lemmas needed to prove Theorem~\ref{t:3_colorable}.
\tTcolorable*

\begin{proof}

    Let $G$ be a 3-colorable co-bridge-free prismatic graph. 

    If $G$ is prime, then the result holds by Lemma~\ref{l:3-col_prime}. Hence we may assume that  $G$ is non prime.

    We prove that either $G$ has a hitting set of size $5$ or it admits a worn 2-chain with both graphs in $\mathcal{Q}_1$.
    To do so, we prove some claims and use them to discuss the worn chain decomposition for $G$ in prime graphs, focusing on graphs in $\mathcal Q _1\cup \mathcal Q _2$. 
    
For the next 4 claims, let $H$ be a $3$-colorable co-bridge-free prismatic graph with $H=H_1\WC H_2$.

    \begin{claim}\label{c:wChain_BicolC4}
    %Let $H$ be a $3$-colorable prismatic graph with $H=H_1\WC H_2$. 
    For $i,j\in\{1,2\}$, $i\neq j$, if $H_i$ contains two disjoint triangles and $H_j$ contains at least one triangle, then $H$ contains a bicolored $C_4$ for any choice of colors.
\end{claim}

\begin{proofclaim}
    Let $(A,B,C)$ be a 3-coloring for $H$ and let $X,Y$ be any two distinct colors in $\{A,B,C\}$.
    
    By Lemma~\ref{l:P3bicolored}, there exists $\mathcal{P}=x_1x_2x_3$, a $P_3$ in $G_i$ with $x_1,x_3\in X_i$ and $x_2\in Y_i$ and $\mathcal{P'}=x'_1x'_2x'_3$, a $P_3$ in $G_i$ with $x'_1,x'_3\in Y_i$ and $x'_2\in X_i$
    
    Let $u\in X_j$ and $v\in Y_j$ be two vertices in a triangle of $H_j$. By Property~\ref{p:color_anticomp}, either $u$ is complete to $X_i$ and $\mathcal{P}\cup \{u\}$ is a bicolored $C_4$ or $v$ is complete to $Y_i$ and $\mathcal{P}'\cup \{v\}$ is a bicolored $C_4$. In both cases the bicolored $C_4$ has vertices in $X\cup Y$. 
\end{proofclaim}

    \begin{claim}\label{c:C4HS1}
    %Let $H$ be a $3$-colorable prismatic graph with $H=H_1\WC H_2$. 
    For $i,j\in\{1,2\}$, $i\neq j$, if $H_i$ contains $\mathcal{C}$, a bicolored $C_4$, then $\Lambda(H_j)\leq 1$. 
    \end{claim}
    
    \begin{proofclaim}
    By Property~\ref{p:color}~(1), there exists a color in $H_j$ that is anticomplete to both colors of $\mathcal{C}$. For the sake of a contradiction, assume that $\Lambda (H_j) \geq 2$. By Property~\ref{p:HS1} and Remark~\ref{r:color_triangles_disjoints} all colors in $H_j$ contain at least two vertices $u,v$. Since all colors are stable sets, $\mathcal{C}\cup \{u,v\}$ yields a co-bridge, a contradiction as $H$ is co-bridge-free. 
    \end{proofclaim}
    
    \begin{claim}\label{c:boundedLambda}
        %Let $H$ be a $3$-colorable prismatic graph with $H=H_1\WC H_2$.For $i,j\in\{1,2\}$, $i\neq j$, if $H_i$ is prime, then the following two implications hold:
        For $i,j\in\{1,2\}$, $i\neq j$, if $H_i$ is prime, then the two following hold:
        \begin{enumerate}
            \item If $\Lambda (H_j) \geq 1$, then $\Lambda(H_i)\leq 4$; 
            \item If $\Lambda (H_j) \geq 2$, then $\Lambda(H_i)\leq 3$.
        \end{enumerate}
    \end{claim}
    
    \begin{proofclaim}
        First note that Lemma~\ref{l:3-col_prime} ensures that $\Lambda(H_i)\leq 5$ as $H_i$ is prime. 
        
        We prove the first implication by contradiction. Therefore we suppose that $\Lambda (H_j)\geq 1$ and $\Lambda(H_i)= 5$. Then $H_i\in\mathcal{Q}_2$ and, by Lemma~\ref{l:canonicallycolored_hitting3ORcycle4}, there exists $\mathcal{C}^+$, a bicolored $C_4+K_1$ in $H_i$. 
        
        By Property~\ref{p:color}~(1), there exists a color in $H_j$ that is anticomplete to both colors of $\mathcal{C^+}$. Moreover, as $\Lambda (H_j) \geq 1$, Property~\ref{p:HS1} and Remark~\ref{r:color_triangles_disjoints} ensure that every color in $H_j$ contains at least one vertex that yields a co-bridge, a contradiction. 
        Therefore $\Lambda (H_i)\leq 4$ and the first implication holds.
        
        We now prove the second implication, also by contradiction. Suppose that $\Lambda (H_j)\geq 2$ and that $\Lambda(H_i)= 4$. Then $H_i\in\mathcal{Q}_2$ and, by Lemma~\ref{l:canonicallycolored_hitting3ORcycle4}, there exists a bicolored $C_4$ in $H_i$. Then, the result follows directly from (\ref{c:C4HS1}).
         % By $\mathcal{W}2$, there exists a color in $H_j$ that is anticomplete to both colors used by the exhibited $C_4$. Moreover, since $\Lambda (H_j) \geq 1$, by Lemma~\ref{p:HS1}, $H_j$ contains at least one triangle and so, all colors in  $H_j$ are nonempty. This yields a co-bridge, a contradiction.
    \end{proofclaim}

    \begin{claim}\label{c:V9HS3}
        %Let $H$ be a $3$-colorable prismatic graph with $H=H_1\WC H_2$. 
        If $|V(H)|\leq 9$, then $\Lambda(H)\leq 3$.
    \end{claim}
    
    \begin{proofclaim}
        Let us consider $H'$ the subgraph induced from $H$ by all vertices that are in at least one triangle in $H$. It is easy to see that $\Lambda(H)=\Lambda(H')$ and that $H'$ is a $3$-colorable co-bridge-free prismatic graph with at most $9$ vertices. Hence, we prove the claim for $H'$. 
        
        If $H'$ is prime then the result follows from the definition if $H'\in \mathcal{Q}_1$ or it arises from Lemma~\ref{l:Path_V9_hitting2V2} if $H'\in\mathcal{Q}_2$.
        
        If $H'$ is not prime, then $H'=H_1\WC H_2$. Since every vertex in  $H'$ belongs to at least one triangle, and since there is no triangle with vertices both in $H_1$ and $H_2$, we know that both $H_1$ and $H_2$ contain at least one triangle. As $9\geq |V(H')|=|V(H_1)|+|V(H_2)|$, we have that $6\geq |V(H_1)|,|V(H_2)|\geq 3$.  
        Observe then that a prismatic graph with $6$ vertices either is a prism or admits a hitting set of size at most $1$. Therefore, if one of $H_1$, $H_2$ has $6$ vertices, since the other has at most $3$ vertices, $\Lambda(H')\leq 3$. If both $H_1$ and $H_2$ have at most $5$ vertices then none of them is a prism and $\Lambda(H')\leq 2$. In both cases the claim holds. 
    \end{proofclaim}

    We now have the intermediate results needed to show the result.
    
    By Theorem~\ref{tcs:prime} (11.1~\cite{ChudAndSey1}), $G$ admits a worn chain decomposition in prime graphs. Set $G=\WC_{i=1}^n G_i$ with $G_i$ being prime $\forall i\in \{1,\dots,n\}$. Observe that $\Lambda (G) = \sum_{i=1}^{n}\Lambda(G_i)$. Moreover, if $G_i\in \mathcal{Q}_0$, its contribution is null. Therefore $\Lambda (G) = \sum_{G_i\in \mathcal{Q}_1\cup\mathcal{Q}_2}\Lambda(G_i)$.

    We argue on the number of graphs in the worn chain decomposition which are in $\mathcal{Q}_1\cup\mathcal{Q}_2$. Observe that all graphs in $\mathcal{Q}_1\cup\mathcal{Q}_2$ contain at least one triangle. 
    
    If there is only one graph of the worn chain decomposition which is in $\mathcal{Q}_1\cup\mathcal{Q}_2$, then the Theorem holds by Lemma~\ref{l:3-col_prime}.
    
    Suppose that the worn chain decomposition for $G$ into prime graphs contains at least two graphs in $\mathcal{Q}_1\cup \mathcal{Q}_2$. Let $i<j$ be the two minimum indexes such that $G_i$ and $G_j$ are in $\mathcal{Q}_1\cup \mathcal{Q}_2$ from a worn chain decomposition for $G$. Set $G^*=\WC_{k=j+1}^n G_k$, which is possibly empty. As previously observed, $G'= G_i\WC G_j \WC G^*$ is an induced subgraph of $G$ and $\Lambda (G)= \Lambda (G') = \Lambda (G_i\WC G_j\WC G^*)=\Lambda (G_i)+\Lambda(G_j)+ \Lambda(G^*)$. As noted in the definition of worn chain decomposition, every triangle in $G'$ belongs to one of the graphs of the decomposition. Therefore it is enough to prove the Theorem for $G'$.

    \vspace{2ex}
      
    First suppose that $\Lambda(G^*)=0$. Hence $\Lambda (G')= \Lambda (G_i)+\Lambda (G_j)$.
    
    We can assume that $\Lambda(G_i),\Lambda(G_j)\leq 3$ for otherwise the Theorem holds by (\ref{c:boundedLambda}). 
    If $\Lambda(G_i)\leq 2$ or $\Lambda(G_j)\leq 2$ then the Theorem holds. Hence we can assume that  $\Lambda(G_i)=\Lambda(G_j)=3$. By Property~\ref{p:HS1} both $G_i$ and $G_j$ contain a prism. 
    By Lemma~\ref{l:prism_worn_V9_NoTriangle}, $|V(G_i)|\leq 9$ and $|V(G_j)|\leq 9$ and by Lemma~\ref{l:Path_V9_hitting2V2}, since $\Lambda(G_i)=\Lambda(G_j)=3$, both $G_i$ and $G_j$ are in $\mathcal{Q}_1$. Therefore, $\WC_{k=1}^{i}G_k$ and $\WC_{k=i+1}^n G_k$ is a worn 2-chain for $G'$ with both graphs containing a $L(K_{3,3})$ and the theorem holds by Lemma~\ref{l:LK33wLK33schlafli}.
        
    \vspace{2ex}

    At this point, we may assume that $\Lambda(G^*)\geq 1$.
     
    Suppose that $\Lambda(G_j)\geq 2$. By (\ref{c:wChain_BicolC4}), both $G_i \WC G_j$ and $G_j\WC G^*$ contain a bicolored $C_4$. Therefore by (\ref{c:C4HS1}), $\Lambda (G_i)\leq 1$ and $\Lambda (G^*)\leq 1$. Hence, we may assume that $\Lambda (G_j)\geq 4$ for otherwise, the Theorem holds. Since $G_j$ is prime, it follows that $G_j$ is a path of triangles graph for otherwise $\Lambda (G_j)\leq 3$. Hence, by Lemma~\ref{l:canonicallycolored_hitting3ORcycle4}, $G_j$ contains a bicolored $C_4$ denoted by $\cal C$. By Property~\ref{p:color} and Remark~\ref{r:color_triangles_disjoints}, there exits a vertex in $G_i$ and a vertex in $G^*$ that are non-adjacent and anticomplete to $\cal C$. This yields a co-bridge, a contradiction.
     
    Suppose now that $\Lambda(G_j)=1$. Observe that both  $G_i \WC G_j$ and $G_j\WC G^*$ contain a prism. By Lemma~\ref{l:prism_worn_V9_NoTriangle}, $|V(G_i)|\leq 9$ and $|V(G^*)|\leq 9$. By (\ref{c:V9HS3}), $|\Lambda (G_i)|\leq 3$ and $|\Lambda (G^*)|\leq 3$. Hence, we may assume that $|\Lambda (G_i)|\geq 2$ and $|\Lambda (G^*)|\geq 2$ for otherwise the Theorem holds. By Property~\ref{p:HS1}, both $G_i$ and $G^*$ contain a prism. Again, by Lemma~\ref{l:prism_worn_V9_NoTriangle} and (\ref{c:V9HS3}) it follows that $|\Lambda (G_i \WC G_j)|\leq 3$ and $|\Lambda (G_j\WC G^*)|\leq 3$. Therefore $|\Lambda (G_i)|= 2$, $|\Lambda (G^*)|=2$ and the theorem holds.

We proved that either $G$ has a hitting set of size at most $5$ or admits a worn $2$-chain with both graphs in $\mathcal{Q}_1$ and a hitting set of size $6$. In the first case, the theorem holds. If $G$ admits a worn $2$-chain with both graphs in $\mathcal{Q}_1$ then, by Lemma~\ref{l:LK33wLK33schlafli}, $G$ is a Schl\"afli-prismatic graph and the Theorem holds.  
\end{proof}

\section{Not $3$-Colorable Orientable Prismatic Graph}\label{s:non3colo}

In the previous section we handled the case of co-bridge-free 3-colorable orientable prismatic graphs and showed that they have a hitting set of size at most $5$ or they have at most 27 vertices. This section is devoted to the analog result for not $3$-colorable orientable prismatic graphs. The proof is based on the structure of the subclass exposed in Section~\ref{s:Struct}. 

\begin{restatable}{theorem}{NonTcol}\label{t:non3col}
   Every co-bridge-free orientable prismatic graph that is not $3$-colorable admits a hitting set of size at most $5$.
\end{restatable}

\begin{proof}
Let $G$ be a co-bridge-free orientable prismatic graph that is not $3$-colorable. By Lemma~\ref{l:CetS_non3col}, either $G$ is not $3$-substantial, or $G$ is a cycle of triangles graph, or $G$ is a ring of five, or $G$ is a mantled $L(K_{3,3})$. 

By definition, if $G$ is not $3$-substantial then $G$ admits a hitting set of size at most $3$. By Lemma~\ref{l:path_cycle_hitting}, if $G$ is a cycle of triangles graph then $G$ has a hitting set of size at most $5$. Hence the remaining cases to handle are when $G$ is a mantled $L(K_{3,3})$ and when $G$ is a ring of five. It will be the aim of the two following claims. 

\begin{claim}
    If $G$ is a mantled $L(K_{3,3})$ then $\Lambda(G)=3$. 
  \end{claim}
  
  \begin{proofclaim}
    Assume that $G$ is a mantled $L(K_{3,3})$ with vertices partitioned into $A, V_1,V_2,V_3,V^1,V^2,V^3$ as in the definition exposed in Section~\ref{s:Struct}.
   
    We first show that every triangle in $G$ is a triangle in $G[A]$. Let $T$ be a triangle in $G$. There is no triangle included in $V_1 \cup V_2\cup V_3$ or $V^1 \cup V^2 \cup V^3$. Since $V_1 \cup V_2\cup V_3$ and $V^1 \cup V^2 \cup V^3$ are anticomplete, we have that $T\cap A\neq \emptyset$. Let $a_i^j\in T$ with $i,j\in \{1,2,3\}$. By definition, $N(a^j_i)\subseteq \{a^{i'}_{j'} : i\neq i', j\neq j'\}\cup V_i\cup V^j$. Observe that $V_i\cup V^j$ is a stable set. Hence there is some $a^{i'}_{j'}$ with $i\neq i', j\neq j'$ such that $a^{i'}_{j'}\in T$. Again, by definition, $a^{i'}_{j'}$ has no neighbor in $V_i\cup V^j$. Hence $T\subset A$. 

    Since $G[A]$ is a $L(K_{3,3})$ and since $\{a_1^1,a_1^2,a_1^3\}$ is a hitting set of $G[A]$ (see Figure~\ref{f:LK33}), $\{a_1^1,a_1^2,a_1^3\}$ is a hitting set of $G$. 
  \end{proofclaim}
  
  \begin{claim}
    If $G$ is a ring of five graph then $\Lambda(G)=3$.
  \end{claim}
  
  \begin{proofclaim}
    Assume that $G$ is a ring of five graph with vertices partitioned into $A,V_0, V_1,V_2,V_3,V_4,V_5$ as in the definition.
   
    We first show that every triangle in $G$ is a triangle in $G[A]$. Let $T$ be a triangle in $G$. Suppose, by contradiction, that there exists $i\in \{0,\dots,5\}$ such that $T\cap V_i\neq \emptyset$. Since $V_i$ is a stable set, $|T\cap V_i|=1$. Since $N(V_0)\subseteq \{b_1,\dots, b_5\}$ and since $\{b_1,\dots, b_5\}$ is a stable set, $i\neq 0$. Hence $N(V_i)\subseteq \{a_{i-1},b_i,a_{i+1}\}\cup V_{i-1}\cup V_{i+1}$. By definition, $\{a_{i-1},b_i,a_{i+1}\}\cup V_{i-1}\cup V_{i+1}$ is a stable set, a contradiction to the existence of $T$. Therefore every triangle in $G$ is a triangle in $G[A]$.

    Since $G[A]$ is a core ring of five and since $\{a_1^1,a_1^2,a_1^3\}$ is a hitting set of $G[A]$ (see Figure~\ref{f:corering5}), $\{a_1^1,a_1^2,a_1^3\}$ is a hitting set of $G$. 
  \end{proofclaim}
  
 Hence, we proved that if $G$ is not 3-colorable then either $G$ is not 3-substantial, or a ring of five graph, or a mantled $L(K_{3,3})$ and $G$ has a hitting set of size 3, or $G$ is or a cycle of triangles graph and has a hitting set of size $5$. 
\end{proof}

\section{Main Results}\label{s:Lafin}
In this section, we proveTheorem~\ref{t:la_total} and explain how to use a polynomial algorithm from the literature to solve the clique covering problem for co-bridge-free prismatic graphs.

Preissmann et al.~\cite{preissmann_complexity_2021} proved the following: 

 \begin{theorem}[Theorem~4.1 in \cite{preissmann_complexity_2021}]\label{l:27_ou_5}
    If $G$ is a non-orientable prismatic graph then $G$ admits a hitting set of cardinality at most $5$ or $G$ is a Schl\"afli-prismatic graph. 
 \end{theorem}
 
 Theorem~\ref{t:la_total} follows directly from theorems formerly exhibited in this article and Theorem~\ref{l:27_ou_5}. 
\Main*

\begin{proof}
By the definition, every co-bridge-free prismatic graph is either orientable or non-orientable. 

If $G$ is non-orientable, then the result holds by Theorem~\ref{l:27_ou_5}.

If $G$ is orientable then either $G$ is 3-colorable and the results holds by Theorem~\ref{t:3_colorable}, or $G$ is not 3-colorable and the results holds by Theorem~\ref{t:non3col}.
\end{proof}

Again, in \cite{preissmann_complexity_2021}, Preissmann et al. showed how the presence of a hitting set with cardinality bounded by a constant can be used to solve the clique cover problem for prismatic graphs. We restate here the main idea and invite the reader to directly rely on Section $4$ in~\cite{preissmann_complexity_2021} for more details. 

Solving the clique cover problem in triangle-free graph consists of computing a matching of maximum cardinality that can be done using Edmonds' algorithm. Preissmann et al. proved the following: 

 \begin{lemma}[Lemma~4.2 in \cite{preissmann_complexity_2021}]\label{supersuper}
 	Let $G$ be a $\{\mbox{diamond}, K_4\}$-free graph.
 	There is an algorithm finding a hitting set of $G$ of cardinality at most $5$ if such a set exists and answering “no" otherwise. 
 	This algorithm has complexity ${\cal O}(n^7)$.
\end{lemma}

A careful reader can note that Preissmann et al. proved in \cite{preissmann_complexity_2021} that, given an integer $k\geq 3$, there is an algorithm in ${\cal O}(n^{k+2})$ that finds a hitting set of $G$ of size at most $k$ if and only if such hitting set exists. If $G$ is a prismatic graph, then it is in particular a graph in the family of $\{\mbox{diamond}, K_4\}$-free graphs, therefore this algorithm applies.  

Let $T(G)$ be the following variant of the adjacency matrix of $G$: for $v, w \in V(G)$, the entry $(v, w)$ of $T(G)$ is $0$ if $v$ and $w$ are not adjacent, $1$ if they are adjacent but without common neighbor, and $x$ if
they are adjacent and have $x$ as a common neighbor. Note that in the last case, if $G$ is diamond-free and $K_4$-free, then $x$ is unique. For any prismatic graphs $G$, $T(G)$ can be computed in ${\cal O}(n^3)$.

Preissmann et al. defined the following method and proved that it provides a minimum clique cover in ${\cal O}(n^{7.5})$: 

\vspace{2ex}

\textbf{Method from \cite{preissmann_complexity_2021}}
\begin{enumerate}
    \item Compute the matrix $T(G)$ as previously defined. This can be
     done in time ${\cal O}(n^3)$.

    \item Use the method from Lemma~\ref{supersuper} in time ${\cal O} (n^{7})$. If the algorithm outputs a hitting set of $G$ of size at most $5$ denoted by $S=\{s_1,\dots,s_{i^*}\}$ ($i^*\leq 5$) then go to Step~\ref{algo:oui}. Else by Theorem \ref{l:27_ou_5}, $G$ is a Schl\"afli-prismatic graph and go to Step~\ref{algo:no}.
     
    \item\label{algo:no} Since there is a bounded number of vertices in $G$ compute all possible clique covers of $G$ in constant time. Go to Step~\ref{algo:fin}.
     
    \item\label{algo:oui} Enumerate all sets of at most $5$ disjoint triangles of $G$. This can trivially be done in time ${\cal O} (n^{15})$ but we can do it in time ${\cal O} (n^{5})$ as follows:

    Compute the set $\mathcal T_i$ of triangles containing $s_i$ for each $1\le i \le i^*$. This can be done in ${\cal O} (n)$ by reading the line of $T(G)$ corresponding to $s_i$. Notice that there are at most $n/2$ triangles in each $\mathcal T_i$. Then compute all subsets $\mathcal T$ of triangles of $G$ obtained by choosing at most one triangle in each $\mathcal T_i$.
    
    For each such $\mathcal T$ which contains only pairwise vertex-disjoint triangles, compute by some classical algorithm a maximum matching $\mathcal M_{\mathcal T}$ of $G\sm (\cup_{T \in \mathcal T} T)$ and let $\mathcal R_{\mathcal T}$ be the vertices of $G$ that are neither in $\mathcal T$ nor in $\mathcal M_{\mathcal T}$. Notice that $\mathcal T\cup \mathcal M_{\mathcal T}\cup \mathcal R_{\mathcal T}$ is a clique cover of $G$. Go to Step~\ref{algo:fin}.

    \item\label{algo:fin} Among all the clique covers generated by the previous steps, let $\mathcal C^*$ be one of the smallest size. Return $\mathcal C^*$.
\end{enumerate}

We leave the reader to check that the proof used in \cite{preissmann_complexity_2021} for non-orientable prismatic graph, only relies on the arguments that the graph is prismatic and either has a hitting set of size at most $5$ or has at most 27 vertices (and is a Schl\"afli-prismatic graph). Therefore the same algorithm can be used for the co-bridge-free prismatic graphs. 

Using the previous remark, a careful reader of \cite{preissmann_complexity_2021} can note that, when restricted to prismatic graphs that admit a hitting set of size at most $k\geq 3$, the clique covering problem can be solved in ${\cal O}(n^{k+2.5})$. 

By bounding the size of the hitting set, we directly obtain that the clique covering problem is polynomial time solvable for co-bridge-free prismatic graphs. 
Therefore the vertex coloring problem for bridge-free antiprismatic graphs can be solved in ${\cal O}(n^{7.5})$.

\section{Conclusion}\label{s:conclusion}
In this article we show that there exists an algorithm that solves the vertex coloring problem for bridge-free antiprismatic graphs in ${\cal O}(n^{7.5})$. This is done by studying the equivalent problem: the clique-covering problem for co-bridge-free prismatic graphs.

We prove that every co-bridge-free prismatic graph admits a hitting set of bounded size or has a bounded number of vertices. Hence the algorithm presented by Preissmann, Robin and Trotignon applies.

The key of this algorithm is the existence of a hitting set of triangles of size bounded by a constant. Preissmann et al. showed that every non-orientable prismatic graphs admits a hitting set of triangles of size at most $5$ or it has at most $27$ vertices. The class of orientable prismatic graphs remains to be discussed. 

By restricting our study to co-bridge-free orientable prismatic graphs and using  Chudnovsky and Seymour's results, we prove that co-bridge-free prismatic graphs also admit a hitting set of size at most $5$ or have at most $27$ vertices. 

Carefully reading Chudnovsky and Seymour's description of this class shows that the difficulty is twofold: it lies in bounding both the size of a worn chain decomposition for 3-colorable orientable prismatic graphs and the size of a hitting set for the families of path and cycle of triangles graphs.

We see two main leads for future work on the complexity of the clique covering problem for prismatic graphs. One direction consists in forbidding other subgraphs for orientable prismatic graphs and show that it also yields bounds on the number of vertex-disjoint triangles. It would extend the subclasses of prismatic graphs for which there is a polynomial algorithm solving the clique covering problem. 
Another approach is to investigate fixed-parameter tractable algorithms with the size of a hitting set of triangles as a parameter. This lead is supported by the fact that Preissmann, Robin and Trotignon proved that the maximum number of disjoint triangles can be found in polynomial time.

\section*{Acknowledgement}
A special thanks to Taite LaGrange for her input on the bound for a hitting set of a path triangles graph during her internship at the Department of Computer Science at Wilfrid Laurier University in 2023. 

This project was partly funded by the GDR-IFM “doctoral student visits” program.
\bibliography{bibli}

\end{document}

%% file: superschema.tex
  %\begin{center}
\begin{tikzpicture}
\begin{scope}[xshift=0cm,yshift=0cm]
		\node[boite] (prismatic) at (0,0.5) {\textbf{Prismatic}};
		
		\node[boite] (nonorient) at (3.6,1) {Non-orientable};
		\node[text width=4cm](papiernous) at (9.5,1){$\Lambda(G)\leq 5$ or Schl\"afli-prismatic graph};
		\draw[->] (nonorient.east)  -> (papiernous.west)         node[midway,fill=white]{\cite{preissmann_complexity_2021}} ;

		\node[boite] (orient) at (3.3,-0.5) {Orientable};
			\node[boite] (non3colorable) at (7.5,-0.5) {Non 3 colorable};
			\draw[->,>=latex] (orient) to[out=0,in=180] (non3colorable);
				\node[fill=white] (lemme) at (7.5,-1.5){Lemma 11.2~\cite{ChudAndSey1}};
				\draw[-] (non3colorable.south) -- (lemme.north);
				
				\node[boite] (Sub) at (10.4,-2.5) {Not $3$-substantial};
				\draw[->,>=latex] (lemme) to[out=-90,in=180] (Sub);
				\node[boite] (R5) at (9.65,-3.5) {Ring of 5};
				\draw[->,>=latex] (lemme) to[out=-90,in=180] (R5);
				\node[boite] (MLK) at (10.25,-4.5) {Mantled $L(K_{3,3})$};
				\draw[->,>=latex] (lemme) to[out=-90,in=180] (MLK);
				\node[boite] (CTG) at (11,-5.5) {Cycle of triangles graph};
				\draw[->,>=latex] (lemme) to[out=-90,in=180] (CTG);
				
				\draw[-] (12.3,-2.5) to[out=0,in=180]  (12.7,-3.5) to[out=180,in=0] (12.3,-4.5);
				\node[text width=2cm](papiernous) at (13.8,-3.5){$\Lambda(G)\leq 3$};
				
			\end{scope}

			\begin{scope}[xshift=-1.5cm,yshift=0.5cm]
			
			\node[boite] (3colorable) at (4.74,-4.25) {3 colorable};
			\draw[->,>=latex] (orient) to[out=-90,in=110] (3colorable);
			
			\node[boite] (nprime) at (2.75,-5.75) {Not Prime};
			\draw[->,>=latex] (3colorable) to[out=-90,in=90] (nprime);
			\node[fill=white,text width=2.25cm] (decomp) at (1,-4.5) {\footnotesize{worn chain decomposition}};
			\draw[-] (nprime) to[out=180,in=-90] (decomp);
			\draw[->,>=latex] (decomp) to[out=75,in=180] (3colorable);

			\node[boite] (prime) at (6.75,-5.75) {Prime};
			\draw[->,>=latex] (3colorable) to[out=-90,in=110]  (prime);
			\node[fill=white] (lemme12) at (7.3,-6.9){Lemma 12.3~\cite{ChudAndSey1}};
			\draw[-] (prime) to[out=-90,in=110]  (lemme12);
			
			\node[boite] (Q0) at (10.65,-8) {$\mathcal{Q}_0$: Triangle-free graphs};
			\draw[->,>=latex] (lemme12) to[out=-90,in=180] (Q0);
			\node[boite] (Q1) at (9.6,-9) {$\mathcal{Q}_1$: $L(K_{3,3})$};
			\draw[->,>=latex] (lemme12) to[out=-90,in=180] (Q1);
			\node[boite] (Q2) at (11,-10) {$\mathcal{Q}_2$: Path of triangles graph};
			\draw[->,>=latex] (lemme12) to[out=-90,in=180] (Q2);
			
			\draw[-] (13.3,-7.8) to[out=0,in=180]  (13.6,-8.5) to[out=180,in=0] (13.3,-9.2);
			\node[text width=2cm](papiernous) at (14.7,-8.5){$\Lambda(G)\leq 3$};

		\draw[->,>=latex] (prismatic) to[out=-10,in=180] (orient);
		\draw[->,>=latex] (prismatic) to[out=10,in=180] (nonorient);
\end{scope}
\end{tikzpicture}
%\end{center}

%% file: LK33.tex
\begin{tikzpicture}
\begin{scope}[xshift=0cm,yshift=0cm]
		\node[labeled] (b2) at (0,0) {$a^2_1$};
		\node[labeled] (c2) at (1,0) {$a^1_3$};
		\node[labeled] (a2) at (0.5,1) {$a^3_2$};
		
		\node[labeled] (b3) at (3,0) {$a^1_2$};
		\node[labeled] (c3) at (4,0) {$a^3_1$};
		\node[labeled] (a3) at (3.5,1) {$a^2_3$};
		
		\node[labeled] (b1) at (1.5,2) {$a^3_3$};
		\node[labeled] (c1) at (2.5,2) {$a^2_2$};
		\node[labeled] (a1) at (2,3) {$a^1_1$};

		\draw[labeled] (a1) -- (b1);
		\draw[labeled] (b1) -- (c1);
		\draw[labeled] (c1) -- (a1);
		
		\draw[-] (a2) -- (b2);
		\draw[-] (b2) -- (c2);
		\draw[-] (c2) -- (a2);
		
		\draw[-] (a3) -- (b3);
		\draw[-] (b3) -- (c3);
		\draw[-] (c3) -- (a3);
		
		\draw[opacity=0.6] (a1) to [out=180,in=135] (a2);
		\draw[opacity=0.6] (a1) to [out=0,in=45] (a3);
		\draw[opacity=0.6] (a2) -- (a3);
		
		\draw[opacity=0.6] (b1) to [out=180,in=135] (b2);
		\draw[opacity=0.6] (b1) -- (b3);
		\draw[opacity=0.6] (b2) to [out=-45,in=-135] (b3);
		
		\draw[opacity=0.6] (c1) to [out=0,in=45] (c3);
		\draw[opacity=0.6] (c1) -- (c2);
		\draw[opacity=0.6] (c2) to [out=-45,in=-135] (c3);	
\end{scope}
\end{tikzpicture}

%% file: CoreRingFive.tex
\begin{tikzpicture}
\begin{scope}[xshift=0cm,yshift=0cm]
		\node[labeled] (a1) at (-1.15,2.3) {$a_1$};
		\node[labeled] (a2) at (1.15,2.3) {$a_2$};
		\node[labeled] (a3) at (1.85,0.25) {$a_3$};
		\node[labeled] (a4) at (0,-0.75) {$a_4$};
		\node[labeled] (a5) at (-1.85,0.25) {$a_5$};
		
		\node[labeled] (b1) at (0.8,0.2) {$b_1$};
		\node[labeled] (b2) at (-0.8,0.2) {$b_2$};
		\node[labeled] (b3) at (-1,1.2) {$b_3$};
		\node[labeled] (b4) at (0,1.7) {$b_4$};
		\node[labeled] (b5) at (1,1.2) {$b_5$};
		\node[fill=white,draw=white,rounded corners](b12) at (0, -1.5) {};

		\draw[labeled] (a1) -- (a2);
		\draw[labeled] (a2) -- (a3);
		\draw[labeled] (a3) -- (a4);
		\draw[labeled] (a4) -- (a5);
		\draw[labeled] (a5) -- (a1);
		
		\draw[labeled] (a1) -- (b1);
		\draw[labeled] (a1) -- (b3);
		\draw[labeled] (a1) -- (b4);

        \draw[labeled] (a2) -- (b2);
		\draw[labeled] (a2) -- (b4);
		\draw[labeled] (a2) -- (b5);
		
		\draw[labeled] (a3) -- (b3);
		\draw[labeled] (a3) -- (b5);
		\draw[labeled] (a3) -- (b1);
		
		\draw[labeled] (a4) -- (b4);
		\draw[labeled] (a4) -- (b1);
		\draw[labeled] (a4) -- (b2);
		
		\draw[labeled] (a5) -- (b5);
		\draw[labeled] (a5) -- (b2);
		\draw[labeled] (a5) -- (b3);
		
\end{scope}
\end{tikzpicture}

%% file: Schlafli.tex
\begin{figure}[H]
 	\centering	
 	\begin{tikzpicture}
 %\node[-] (R) at (1,2.5) {R};
 \node[-] (r11) at (0,2) {$r_1^1$};
 \node[-] (r12) at (0,1) {$r_1^2$};
 \node[-] (r13) at (0,0) {$r_1^3$};
 \node[-] (r21) at (1,2) {$r_2^1$};
 \node[-] (r22) at (1,1) {$r_2^2$};
 \node[-] (r23) at (1,0) {$r_2^3$};
 \node[-] (r31) at (2,2) {$r_3^1$};
 \node[-] (r32) at (2,1) {$r_3^2$};
 \node[-] (r33) at (2,0) {$r_3^3$};
 
 %\node[-] (S) at (5,-0.5) {S};
 \node[-] (s11) at (3,-1) {$s_1^1$};
 \node[-] (s12) at (3,-2) {$s_1^2$};
 \node[-] (s13) at (3,-3) {$s_1^3$};
 \node[-] (s21) at (4,-1) {$s_2^1$};
 \node[-] (s22) at (4,-2) {$s_2^2$};
 \node[-] (s23) at (4,-3) {$s_2^3$};
 \node[-] (s31) at (5,-1) {$s_3^1$};
 \node[-] (s32) at (5,-2) {$s_3^2$};
 \node[-] (s33) at (5,-3) {$s_3^3$};
 
 %\node[-] (T) at (-3.5,-4) {T};
 \node[-] (t11) at (-3,-2) {$t_1^1$};
 \node[-] (t12) at (-3,-3) {$t_1^2$};
 \node[-] (t13) at (-3,-4) {$t_1^3$};
 \node[-] (t21) at (-2,-2) {$t_2^1$};
 \node[-] (t22) at (-2,-3) {$t_2^2$};
 \node[-] (t23) at (-2,-4) {$t_2^3$};
 \node[-] (t31) at (-1,-2) {$t_3^1$};
 \node[-] (t32) at (-1,-3) {$t_3^2$};
 \node[-] (t33) at (-1,-4) {$t_3^3$};

 \draw[-] (r33) -- (r22); 
 \draw[-] (r33) -- (r21); 
 \draw[-] (r11) to[bend left=78] (r33); 
 
 \draw[-] (r11) -- (r22); 
 \draw[-] (r12) -- (r21); 
 \draw[-] (r12) -- (r33);

  \draw[-] (r33) to[bend left=80](s13); 
  \draw[-] (r33) to[bend left=80] (s23); `
  \draw[-] (r33) to[bend left=80] (s33); 
  
   \draw[-] (r33) -- (t31); 
   \draw[-] (r33) -- (t32); 
   \draw[-] (r33) -- (t33); 
   
   \draw[-] (s13) to[bend left=50] (t31); 
   \draw[-] (s23) to[bend left=50] (t32); 
   \draw[-] (s33) to[bend left=50] (t33); 
 \end{tikzpicture}
 \caption{ The complement of the Schläfli graph where only the 10 edges incident to $r_3^3$  and the 5 triangles that contain $r_3^3$ are represented.\label{f:schafli}}
 \end{figure}	